%% file: main.tex
\documentclass[11pt]{article}

\usepackage{amsmath}
\usepackage{amssymb}
\usepackage{amsthm}
\usepackage{algorithm}
\usepackage{algpseudocode}
\usepackage[english]{babel}
\usepackage{enumitem}
\usepackage{fullpage}
\usepackage{graphicx}
\usepackage{hyperref}
\usepackage{microtype}
\usepackage{mathrsfs}
\usepackage{mathtools}
\usepackage[nice]{nicefrac}
\usepackage{subcaption}
\usepackage{xspace}

\DeclareMathOperator*{\argmin}{argmin}

\theoremstyle{plain}
\newtheorem{theorem}{Theorem}[section]

\newtheorem{lemma}[theorem]{Lemma}
\newtheorem{claim}[theorem]{Claim}
\newtheorem{corollary}[theorem]{Corollary}
\theoremstyle{definition}

\theoremstyle{remark}

\usepackage{authblk}

\title{Minimum Congestion Routing of Unsplittable Flows \\ in Data-Center Networks} 

\author[1, 2, 3]{Miguel Ferreira}
\author[1]{Nirav Atre}
\author[1]{Justine Sherry}
\author[4]{Michael Dinitz}
\author[2, 3]{João Luís Sobrinho}

\affil[1]{Carnegie Mellon University}
\affil[2]{Instituto de Telecomunicações}
\affil[3]{Instituto Superior Técnico, Universidade de Lisboa}
\affil[4]{John Hopkins University}

\begin{document}

\maketitle

\begin{abstract}
    \input{0-abstract}

\end{abstract}

\clearpage

\tableofcontents

\clearpage

\pagenumbering{arabic}

\section{Introduction}
\label{sec:introduction}

\input{1-introduction}

\section{Summary of Results}
\label{sec:summary_results}
\input{2-summary_results}

\section{Problem Setting}
\label{sec:problem_setting}
\input{3-problem_setting}

\section{Offline: A New Routing Algorithm}
\label{sec:main_result}
\input{4-main_result}

\section{Offline: Lower Bounds on Congestion and Approximation}
\label{sec:offline-lower-bounds}
\input{5-offline_lower_bounds}

\section{Online: Lower Bounds on Congestion and Approximation}
\label{sec:online-lower-bounds}
\input{6-online_lower_bounds}

\section{Related Work}
\label{sec:related_work}
\input{7-related_work}

\section{Discussion and Open Questions}
\label{sec:discussion}
\input{8-discussion}

\bibliographystyle{alpha}

\newcommand{\etalchar}[1]{$^{#1}$}

\end{document}

%% file: 0-abstract.tex
Millions of flows are routed concurrently through a modern data-center. These networks are often built as Clos topologies, and flow demands are constrained only by the link capacities at the ingress and egress points. The minimum congestion routing problem seeks to route a set of flows through a data center while minimizing the maximum flow demand on any link. This is easily achieved by splitting flow demands along all available paths. However, arbitrary flow splitting is unrealistic. Instead, network operators rely on heuristics for routing unsplittable flows, the best of which results in a worst-case congestion of $2$ (twice the uniform link capacities). But is $2$ the lowest possible congestion? If not, can an efficient routing algorithm attain congestion below $2$?

Guided by these questions, we investigate the minimum congestion routing problem in Clos networks with unsplittable flows. First, we show that for some sets of flows the minimum congestion is at least $\nicefrac{3}{2}$, and that it is $NP$-hard to approximate a minimum congestion routing by a factor less than $\nicefrac{3}{2}$. Second, addressing the motivating questions directly, we present a polynomial-time algorithm that guarantees a congestion of at most $\nicefrac{9}{5}$ for any set of flows, while also providing a $\nicefrac{9}{5}$ approximation of a minimum congestion routing. Last, shifting to the online setting, we demonstrate that no online algorithm (even randomized) can approximate a minimum congestion routing by a factor less than $2$, providing a strict separation between the online and the offline setting.

%% file: 1-introduction.tex
\quad Today's data-centers are asked to route millions of \textit{flows} simultaneously~\cite{Benson_2010}. 
\unskip\footnote{Each \textit{flow} corresponds a network connection between a source-destination pair, translating to a commodity in network flow problems.}
Each flow is offered with a demand limited by the capacities of links leaving and entering servers, and the routing of these flows seeks to minimize \textit{congestion}~\cite{Alizadeh_2014, Singla_2014, Chiesa_2017, Namyar_2021}. The congestion of a link is the ratio between the total demand routed through the link and the link capacity, and the congestion of a routing is the maximum congestion over all internal~links. If the congestion of a routing is less than $1$, then the demand of every flow is satisfied; otherwise, some flows only obtain a fraction of their demands given by the inverse of the congestion.

\textit{Clos networks} are the de facto topologies underlying most of today's data-centers~\cite{Al-Fares_2008, Greenberg_2009, Roy_2015, Singh_2015, Qureshi_2022, Gangidi_2024, Qian_2024}. In a Clos network, each source (destination) server is linked to a single input (output) top-of-rack (ToR) switch, and each ToR switch is linked to as many middle switches as there are servers per ToR switch~\cite{Clos_1953}. Links have uniform capacities (see Figure~\ref{fig:data-center}). A basic property of Clos networks is their \textit{full bisection bandwidth}, meaning that the capacity of any cut separating all sources from all destinations is at least the aggregate capacity over all links between sources and input switches (destinations and output switches). This property has a notable implication for routing in Clos networks, which motivated their widespread adoption in current data-centers: uniformly splitting flow demand over all source-destination paths yields a routing with congestion at most $1$~\cite{Chiesa_2017}. 

However, in the context of modern data-centers, flows are often \textit{unsplittable}, since forwarding data-packets concurrently on multiple paths requires extensive modifications to current transport protocols, which are difficult to deploy in practice~\cite{Qureshi_2022, Gangidi_2024, Qian_2024}. Therefore, we ask the following fundamental questions concerning minimum congestion routings of unsplittable flows and their efficient computation in Clos networks:

\begin{enumerate}[label=(\textbf{Q\arabic*)}]
    \item \textbf{How close to $\mathbf{1}$ is the congestion of a minimum congestion routing of unsplittable flows?} 
    
    \item \textbf{How well can a minimum congestion routing of unsplittable flows be \linebreak approximated by a polynomial-time algorithm?} 
\end{enumerate}

Prior work offers only tentative answers to these questions. In general networks, the theory community has introduced algorithms for routing unsplittable flows based on randomized rounding of the multi-commodity flow relaxation, and shown that it is possible to approximate a minimum congestion routing within logarithmic factors (in the size of the network)~\cite{Raghavan_1987, Kolman_2006, Chakrabarti_2007}, but that it is NP-hard to approximate it within constant factors~\cite{Chuzhoy_2007}. While these results hold for general network, the symmetric structure of Clos networks allows for better bounds. 

\begin{figure}
    \centering
    \captionsetup{labelfont=bf}
    \begin{minipage}[t]{.49\textwidth}
        \centering
        \includegraphics[width=0.85\textwidth]{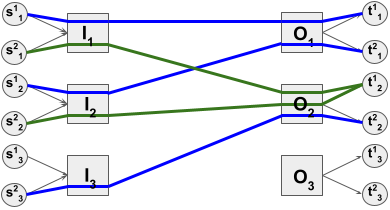} 
        \subcaption{A set of flows without a routing.}
        \label{fig:macro}
    \end{minipage}
    \begin{minipage}[t]{.49\textwidth}
        \centering
        \includegraphics[width=0.85\textwidth]{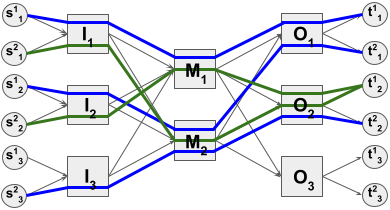} 
        \subcaption{A routing with minimum congestion.}
        \label{fig:clos}
    \end{minipage}
    \caption{The minimum congestion routing problem in Clos networks. Circles symbolize servers, and squares symbolize switches. Lines symbolize flows from source to destination. Flows are represented without and with routing, with middle switches omitted from the former. Figure~\ref{fig:macro} shows a set of flows in a Clos network with $2$ middle switches and $3$ input (output) ToR switches: flows $(I_1, O_1)$, $(I_2, O_1)$, and $(I_3, O_2)$ have demand $1$ (in blue), and flows $(I_1, O_2)$ and $(I_2, O_2)$ have demand $\nicefrac{1}{2}$ (in green). Figure~\ref{fig:clos} shows a minimum congestion routing. Every link except for $M_2 O_2$ is traversed by at most one flow, and thus has congestion at most $1$. Link $M_2 O_2$ is traversed by flows $(I_3, O_2)$ and $(I_2, O_2)$, and thus has congestion $\nicefrac{3}{2}$. Hence, the routing has congestion $\nicefrac{3}{2}$.}
    \label{fig:data-center}
\end{figure}

In Clos networks, the networking community has proposed heuristics, and shown their good empirical performance compared to uniformly splitting flows. The long-standing heuristic is Equal-Cost-Multi-Path (ECMP)~\cite{Al-Fares_2008}, which gives logarithmic worst-case congestion and approximation factor. ECMP assigns each flow to a path chosen uniformly~at~random. State-of-the-art heuristics~\cite{Melen_1989, Al-Fares_2010, Curtis_2011, Alizadeh_2014, Chiesa_2017, Hsu_2020} are grounded in two algorithms, both of which give worst-case congestion and approximation factor of $2$. The \textit{Sorted Greedy algorithm}~\cite{Al-Fares_2010} sorts the flows in decreasing order of demands and assigns each flow to a path whose congestion is minimum; the congestion of a path is the maximum congestion over its links (see Figure~\ref{fig:intro-online_lower_bound_congestion}). The \textit{Melen-Turner algorithm}~\cite{Melen_1989}) first constructs a new Clos network from the original one, and maps the flows to the new network. In the new network there are multiple copies of each original ToR switch, with flows assigned to copies of its ToR switches such that there are at most as many flows per copy as there are middle switches and demands decrease with the indexing~of the copies. Then, the algorithm finds a link-disjoint routing for the flows in the new network, thus obtaining a routing for them in the original one (see Figure~\ref{fig:worst-case-melen-turner-algorithm}).

In view of prior work, questions \textbf{Q1} and \textbf{Q2} can be restated as follows. Is $2$ the tight factor for congestion and approximation, in which case known heuristics are optimal? Or does the structure of Clos networks allow for the existence of routings with congestion close to $1$ and their polynomial-time computation?

\begin{figure}[t!]
    \centering
    \captionsetup{labelfont=bf}
    \begin{subfigure}{.49\textwidth}
        \includegraphics[width=0.96\textwidth]{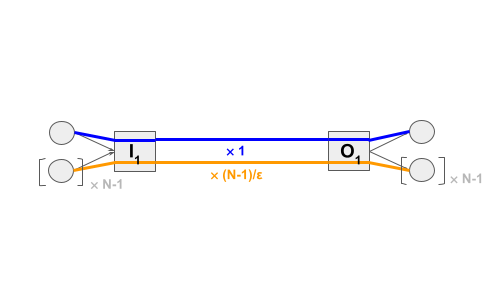} 
        \subcaption{A set of flows in the original network.}
        \label{fig:worst-case-melen-turner-algorithm-1}
    \end{subfigure}
    \begin{subfigure}{.49\textwidth}
        \includegraphics[width=0.96\textwidth]{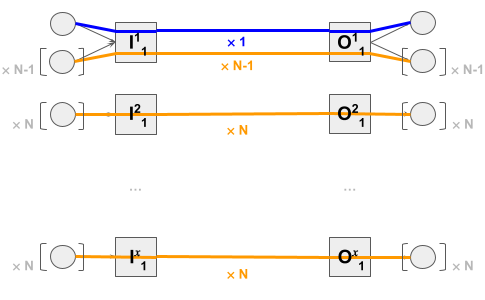} 
        \subcaption{The set of flows mapped to the new network.}
        \label{fig:worst-case-melen-turner-algorithm-2}
    \end{subfigure}

    \vspace{10pt}

    \begin{subfigure}{.49\textwidth}
        \includegraphics[width=0.96\textwidth]{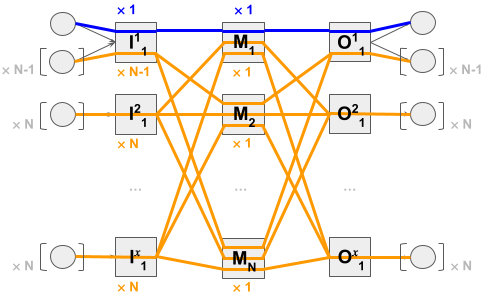} 
        \subcaption{Link-disjoint routing in the new network.}
        \label{fig:worst-case-melen-turner-algorithm-3}
    \end{subfigure}
    \begin{subfigure}{.49\textwidth}
        \includegraphics[width=0.96\textwidth]{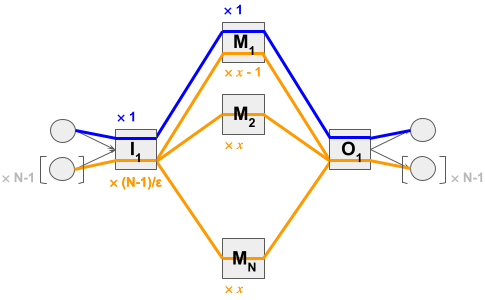} 
        \subcaption{Corresponding routing in the original network.}
        \label{fig:worst-case-melen-turner-algorithm-4}
    \end{subfigure}
    \caption{Worst-case flows for the Melen-Turner algorithm. Figure~\ref{fig:worst-case-melen-turner-algorithm-1} shows a set of flows in~a~Clos network with $N$ middle switches composed of two types of flows: one type 1 flow with demand $1$ (in blue) and $\nicefrac{(N-1)}{\epsilon}$ type 2 flows each with demand $\epsilon$ (in orange), for some small $\epsilon > 0$ (only the first input and output switches are shown). The minimum congestion routing has congestion $1$, with the type 1 flow assigned to some middle switch, and $\nicefrac{1}{\epsilon}$ type 2 flows assigned to each of the remaining $N-1$. Figure~\ref{fig:worst-case-melen-turner-algorithm-2} shows the mapping of the flows to the new Clos network, where each ToR switch is divided into $x$ copies, with $x \coloneqq \nicefrac{1}{N} + \nicefrac{1}{\epsilon} - \nicefrac{1}{(N \times \epsilon)}$. Figure~\ref{fig:worst-case-melen-turner-algorithm-3} shows a link-disjoint routing for the flows in the new network, and Figure~\ref{fig:worst-case-melen-turner-algorithm-4} shows the corresponding routing in the original one. The routing returned by the algorithm has congestion $2 - \epsilon - \nicefrac{(1-\epsilon)}{N}$. Consequently, the worst-case congestion and the approximation factor of the algorithm approach $2$ as $N$ grows large.}
    \label{fig:worst-case-melen-turner-algorithm}
\end{figure}

%% file: 2-summary_results.tex
\quad We introduce new algorithms and bounds for the minimum congestion routing problem in Clos networks with unsplittable flows in two settings: the \textit{offline setting}, where a set of flows is offered to the data-center at once and can be routed altogether, and the \textit{online setting}, where a sequence of flows is offered to the data-center one at a time, and each flow must be routed before receiving the next flow and without re-routing the previously routed flows. \\

\textbf{Main Result -  Offline: A New Routing Algorithm.} We design a routing algorithm that breaks through the barrier of $2$ established by known heuristics. The new algorithm routes a set of flows in two phases bridged by a threshold on the congestion of each phase. The first phase routes a subset of the flows via the Melen-Turner algorithm with congestion at most the threshold. The second phase routes the remaining flows via the Sorted Greedy algorithm without increasing congestion beyond the threshold. While the worst-case congestion of the algorithms underlying each phase is $2$, by setting the threshold to $\nicefrac{9}{5}$ the new algorithm mitigates the worst-case flows of both, thus improving the worst-case congestion to $\nicefrac{9}{5}$.    

In the present form, the algorithm returns a routing with congestion at most $\nicefrac{9}{5}$, but fails to approximate a minimum congestion routing by a factor at most $\nicefrac{9}{5}$, since the minimum congestion could be less than $1$. If the minimum congestion were known \textit{apriori}, then the algorithm could be easily fixed by setting the threshold to $\nicefrac{9}{5}$ times the minimum congestion; given that it is not, a more nuanced bridging of the two phases is required. In the final form, the first phase of the algorithm routes a subset of the flows that includes all those with demand at least $\nicefrac{1}{3}$ times the minimum with congestion at most $\nicefrac{9}{5}$ times a suitable lower bound on the minimum, and the second phase routes the remaining flows without increasing congestion beyond $\nicefrac{9}{5}$ times the minimum.

\begin{theorem}
	There is a polynomial-time algorithm that returns a routing with congestion at most $\nicefrac{9}{5}$ and approximates a minimum congestion routing by a factor at most $\nicefrac{9}{5}$.
\end{theorem}



\textbf{Offline: Lower Bounds on Congestion and Approximation.} We present lower bounds that show that, despite the special structure of Clos networks, a minimum congestion routing cannot have worst-case congestion less than $\nicefrac{3}{2}$, nor be approximated by a factor less than $\nicefrac{3}{2}$. Another basic property of Clos networks is that, if flow demands are unit, implying that sources and destinations are restricted to at most one flow, then the congestion of a minimum congestion routing is at most $1$, and it is possible to find such a routing in polynomial-time~\cite{Hwang_1983, Lovasz_2009}. This follows from the equivalence between finding link-disjoint routings in Clos networks and decomposing bounded degree bipartite graphs into as many matchings. We show that, as soon as this premise is relaxed to allow half-unit flow demands, the rise of new packing constraints entails that the previous property no longer holds true.

\begin{theorem}
    \label{thm:intro-offline_limit_congestion}
    There is a set of flows for which the congestion of a minimum congestion routing is at least $\nicefrac{3}{2}$.
\end{theorem}

\begin{theorem}
    \label{thm:intro-offline_limit_approximation}
    For a set of flows with unit or half-unit demands, deciding whether there is a routing with congestion at most $1$ is NP-complete.
\end{theorem}

The proofs of both theorems make use of a special set of flows with unit demands that we call a \textit{cross gadget}, which in a Clos network with $N$ middle switches maps the sources of $N$ input switches, $N-1$ sources per input switch, to the destinations of $N-1$ output switches, $N$ destinations per output switch (see Figure~\ref{fig:cross-gadget}). The key property of the cross gadget is that in a routing with congestion $1$ there is a different middle switch at each input switch to which no unit demand flow is assigned; consequently, any additional flows leaving each of these input switches are assigned to different middle switches. While in Theorem~\ref{thm:intro-offline_limit_congestion} the cross gadget is used to design a set of flows for which no routing with congestion $1$ exists, in Theorem~\ref{thm:intro-offline_limit_approximation} it is used to establish a reduction from the $3$-edge coloring problem. The second theorem implies a hardness of approximation result.

\begin{corollary}
    No polynomial-time algorithm can approximate the congestion of a minimum congestion routing by a factor less than $\nicefrac{3}{2}$ unless $P = NP$.
\end{corollary}

\begin{figure}[t!]
    \centering
    \captionsetup{labelfont=bf}
    \begin{subfigure}{.49\textwidth}
        \centering
        \captionsetup{labelfont=bf}
        \includegraphics[width=0.85\textwidth]{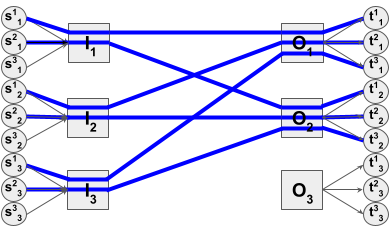} 
        \subcaption{The cross gadget without routing.}
    \end{subfigure}
    \begin{subfigure}{.49\textwidth}
        \centering
        \captionsetup{labelfont=bf}
        \includegraphics[width=0.85\textwidth]{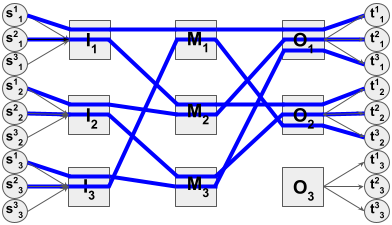} 
        \subcaption{A routing for the gadget with congestion~$1$.}
    \end{subfigure}
    \caption{The cross gadget underlying Theorems~\ref{thm:intro-offline_limit_congestion}  and \ref{thm:intro-offline_limit_approximation} in a Clos network with $3$ middle switches.}
    \label{fig:cross-gadget}
\end{figure}

\textbf{Online: Lower Bounds on Congestion and Approximation.} In practice, it is important to also be able to route flows \emph{online}. In this context, we prove a separation between the offline and online settings. While in the offline setting the new algorithm has worst-case congestion less than $2$, in the online setting we show that no deterministic algorithm can ensure this property even if flow demands are unitary. Moreover, we show that randomizing the routing choices does not help. 

\begin{theorem}
    \label{thm:intro-online}
    For every online algorithm, deterministic or randomized, there is a sequence of flows with unit demands for which the congestion of the algorithm is at least $2$.
\end{theorem}

\begin{figure}[t!]
    \centering
    \begin{minipage}[t]{.49\textwidth}
        \centering
        \includegraphics[width=0.8\textwidth]{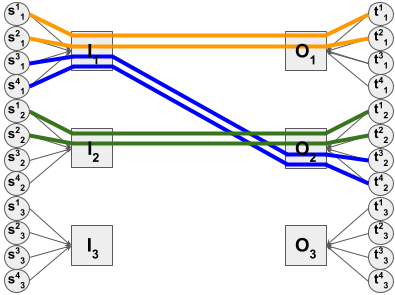} 
        \subcaption{Sequence $X$.}
        \label{fig:intro-online_lower_bound_congestion_1}
    \end{minipage}
    \begin{minipage}[t]{.49\textwidth}
        \centering
        \includegraphics[width=0.8\textwidth]{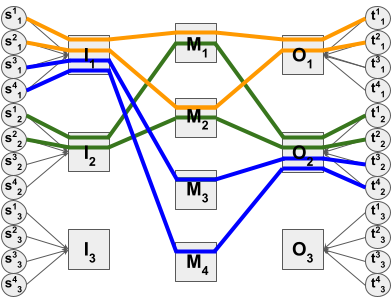} 
        \subcaption{Congestion $1$ routing of $X$.}
        \label{fig:intro-online_lower_bound_congestion_2}
    \end{minipage}
    
    \vspace{5pt}
    
    \begin{subfigure}{.49\textwidth}
        \centering
        \captionsetup{labelfont=bf}
        \includegraphics[width=0.8\textwidth]{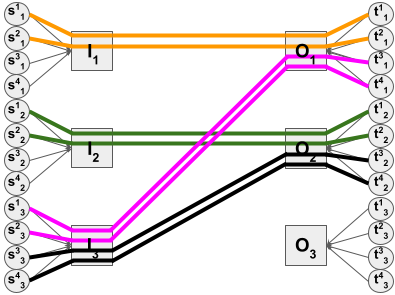} 
        \subcaption{Sequence $Y$.}
        \label{fig:intro-online_lower_bound_congestion_3}
    \end{subfigure}
    \begin{subfigure}{.49\textwidth}
        \centering
        \captionsetup{labelfont=bf}
        \includegraphics[width=0.8\textwidth]{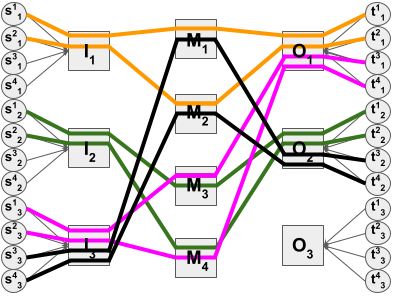} 
        \subcaption{Congestion $1$ routing of $Y$.}
        \label{fig:intro-online_lower_bound_congestion_4}
    \end{subfigure}
    
    \caption{The sequences of flows underlying Theorem~\ref{thm:intro-online} in a Clos network with $3$ middle switches, from which worst-case flows for the Sorted Greedy algorithm can be derived. Figures~\ref{fig:intro-online_lower_bound_congestion_1}~and \ref{fig:intro-online_lower_bound_congestion_2} show the two sequences denoted by $X = (X_1, X_2)$ and $Y = (Y_1, Y_2)$, with $X_1 = Y_1$ and $X_2 \neq Y_2$. Subsequence $X_1$ consists of two flows $(I_1, O_1)$ and two flows $(I_2, O_2)$, $X_2$ of two flows $(I_1, O_2)$, and $Y_2$ of two flows $(I_3, O_1)$ and two flows $(I_3, O_2)$. Figures~\ref{fig:intro-online_lower_bound_congestion_3} and~\ref{fig:intro-online_lower_bound_congestion_4} show the unique routing for each sequence with congestion $1$. If flows in $X_1$ have demand $1 - \epsilon$ and those in $X_2$ and $Y_2$ have demand $1$, then the congestion of the routing returned by the Sorted Greedy algorithm is $2-\epsilon$ for at least one of the sequences depending on the tie-break. Consequently, both the worst-case congestion and the approximation factor of the algorithm are $2-\epsilon$, which approach $2$ as $\epsilon$ shrinks small.}
    \label{fig:intro-online_lower_bound_congestion}
\end{figure}

The proof of the theorem is divided into two parts, the first showing the result for deterministic algorithms, and the second generalizing it for randomized algorithms. The~first part designs two sequences of flows that agree in their prefixes but disagree in their suffixes (see Figure~\ref{fig:intro-online_lower_bound_congestion}). The key property of the sequences is that in a routing with congestion $1$ the routing of their common prefix differs. Consequently, any deterministic algorithm that returns a routing with congestion $1$ for one sequence does not return it for the other. The second part designs $2^S$ supersequences consisting of $S$ independent sequences, each of which corresponds to one of the previous sequences, for some $S$ linear in the number of ToR switches. Therefore, from the key property of the sequences, any deterministic algorithm fails to return a routing with congestion $1$ for at least $2^S-1$ supersequences; thus, the expected congestion of an optimal deterministic algorithm when supersequences are chosen uniformly at random over the $2^S$ supersequences is at least $2-\nicefrac{1}{2^s}$, which approaches $2$ as $S$ grows large. Then, the theorem follows from the application of Yao's Minimax Principle~\cite{Yao_1977, Motwani_1995}. 


Since with unitary flow demand there is a routing with congestion at most $1$~\cite{Hwang_1983, Lovasz_2009}, the theorems imply a lower bound on the approximation factor of any online algorithm.

\begin{corollary}
    No online algorithm, deterministic or randomized, can approximate the congestion of a minimum congestion routing by a factor less than $2$.
\end{corollary}

In contrast to these lower bounds, it is not hard to show that the \textit{Unsorted Greedy algorithm} gives an upper bound of $3$ on congestion and approximation.

\subsection{Roadmap}

\quad In \S\ref{sec:problem_setting}, we present a formal description of the minimum congestion routing problem in Clos networks. In \S\ref{sec:main_result}, we design and analyze the new routing algorithm for the minimum congestion routing problem. In \S\ref{sec:offline-lower-bounds}  and~\S\ref{sec:online-lower-bounds} we prove the lower bounds on congestion and approximation of a minimum congestion routing in the offline and online settings, respectively. In \S\ref{sec:related_work}, we review related work, and, in \S\ref{sec:discussion}, we discuss open questions.

%% file: 3-problem_setting.tex
\quad We formalize the minimizing congestion routing problem in Clos networks. In \S\ref{sec:problem_setting:network}, we describe Clos networks, and, in \S\ref{sec:problem_setting:congestion}, we define a minimum congestion~routing.

\subsection{Clos Networks}
\label{sec:problem_setting:network}

\quad A Clos network interconnects \textit{source servers} to \textit{destination servers} via three stages of switches: \textit{input ToR switches}, \textit{middle switches}, and \textit{output ToR switches}. We assume that the number of servers per ToR switch equals the number of middle switches (so that the Clos network has full bisection bandwidth). The Clos network of dimension $(N, R)$, for positive integers $N$ and $R$, is denoted by $C_{N, R}$ and was illustrated in Figure~\ref{fig:clos} for $N = 2$ and $R = 3$, and designates the directed graph with the following vertex and edge sets:

\begin{itemize}
    \item \textbf{Vertex set:} There are $N$ middle switches, with the $m$'th middle switch denoted by $M_m$, $m \in \left[ N \right]$. There are $R$ input switches, with the $i$'th input switch denoted by $I_i$, and $R$ output switches, with the $i$'th output switch denoted by $O_i$, $i \in \left[ R \right]$. There are $N \times R$ source servers, with the $k$'th source server of the $i$'th input switch denoted by $s^k_i$, and $N \times R$ destination servers, with the $k$'th destination server of the $i$'th output switch denoted by $t^k_i$, $i \in \left[ R \right]$ and $k \in \left[ N \right]$.
    \item \textbf{Edge set:} There is one edge $I_i M_m$ and one edge $M_m O_i$, $i \in \left[ R \right]$ and $m \in \left[ N \right]$. There is one edge $s^k_i I_i$ and one edge $O_i t^k_i$, $i \in \left[ R \right]$ and $k \in \left[ N \right]$.
\end{itemize}

In $C_{N, R}$, there are $N$ servers per ToR switch; the radices of ToR and middle switches are $N$ and $R$, respectively; and there are $N$ paths between every source-destination pair, each path via a different middle switch. We assume that all links have capacity $1$.

\subsection{Minimum Congestion Routing}
\label{sec:problem_setting:congestion}

\quad \enskip The input to the minimum congestion routing problem in a Clos network is: (1) A Clos network $C_{N, R}$. (2) A set $\mathcal{F}$ of flows. Each flow $f$ maps to a source-destination server pair, and a positive \textit{demand} $\text{dem}(f)$. If clear from the context, then a flow may be identified by its input-output switch pair. Given a flow $f \in \mathcal{F}$, let $i(f)$ and $s(f)$ be the input switch and the source server that $f$ leaves, respectively, and $j(f)$ and $t(f)$ be the output switch and the destination server that $f$ enters, respectively. The demands satisfy the following assumption: the total demand over all flows leaving a source and entering a destination is at most $1$ (so that if flows were splittable rather than unsplittable there were a routing with congestion $1$ for every set of flows):
\unskip\footnote{This corresponds to assuming that the demands are drawn from the \textit{hose-model} of a Clos network~\cite{Duffield_1999, Singla_2014, Namyar_2021}.}
\begin{align*}
    \displaystyle\sum\limits_{\substack{ f \in \mathcal{F}: i(f) = i \\ \hspace{23pt} s(f) = k}} \text{dem}(f) \le 1 \quad \text{ and }
    \displaystyle\sum\limits_{\substack{ f \in \mathcal{F}: j(f) = i \\ \hspace{23pt} t(f) = k}} \text{dem}(f) \le 1 \text{ for all } i \in [R], k \in [N].
\end{align*}

A solution to the problem is a \textit{routing} $r$ for the set $\mathcal{F}$ of flows, which is an assignment from each flow $f$ to a source-destination path $r(f)$; in fact, since there is a bijection between source-destination paths and middle switches, a routing is an assignment to a middle switch. An optimal solution to the problem minimizes over all routings the congestion of each routing; we call it a \textit{minimum congestion routing}. The congestion of a routing is the maximum congestion over all links between ToR and middle switches, where the congestion of a link is the total demand over all flows traversing the link. Letting $\mathcal{R}$ denote the set of all routings, the congestion of a minimum congestion routing, denoted by $OPT( \mathcal{F} )$, writes:
\begin{align*}
    OPT( \mathcal{F} ) \coloneq 
    \min_{r \in \mathcal{R}}
    \underbrace{
    \max_{\substack{ i \in [R] \\ m \in [N] }}
    \max \Big\{
    \overbrace{
    \displaystyle\sum\limits_{\substack{ f \in \mathcal{F}: i(f) = i \\ \hspace{26pt} r(f) = m}} \text{dem}(f)}^{\text{congestion of $I_i M_m$}},
    \overbrace{
    \displaystyle\sum\limits_{\substack{ f \in \mathcal{F}: j(f) = i \\ \hspace{26pt} r(f) = m}} \text{dem}(f) \Big\}}^{\text{congestion of $M_mO_j$}}}_{\text{congestion of $r$}}
\end{align*}
If clear from the context, then the argument $\mathcal{F}$ is omitted. The minimum congestion routing problem was illustrated in Figure~\ref{fig:data-center}.

While the assumption that the total demand leaving a source and entering a destination is at most $1$ limits the worst-case congestion of a minimum congestion routing, it implies no loss of generality concerning its approximation, in the sense that if there is a $p$-approximation algorithm for the problem when the demands satisfy the previous assumption, then there is also a $p$-approximation algorithm for the problem when the demands do not. (The easy proof is omitted.) \newline 


The next theorem recalls a basic property of Clos networks concerning the special case where there is at most one flow per source and per destination. We refer to this theorem and its corollary repeatedly throughput the paper.

\begin{theorem}[\cite{Hwang_1983, Lovasz_2009}]
    \label{thm:rearrangeability}
    Consider a Clos network $C_{N, R}$. For every set of flows such that the number of flows per source and per destination is at most one, there is a link-disjoint routing of the flows, and one such a routing can be found in polynomial-time.
\end{theorem}

The proof of the theorem establishes a bijection between a set of flows in a Clos network with at most $N$ flows per ToR switch and a bipartite graph with vertex degree at most $N$, and decomposes the bipartite graph into $N$ matchings, each matching corresponding to the subset of flows assigned to a different middle switch.

\begin{corollary}
    \label{thm:rearrangeability_2}
    Consider a Clos network $C_{N, R}$. For every set of flows such that the number of flows per source and per destination is at most one, the congestion of a minimum congestion routing of the flows is less than or equal to $1$, and one such a routing can be found in polynomial time. 
\end{corollary}

%% file: 4-main_result.tex
\quad We introduce a new algorithm for minimum congestion routing in Clos networks. In \S\ref{sec:design_algorithm}, we design the algorithm, and, in \S\ref{sec:analysis_algorithm}, we analyze it to show that the algorithm approximates a minimum congestion routing by a factor at most $\nicefrac{9}{5}$. A straightforward generalization of the analysis shows that the algorithm also returns a routing with congestion at most $\nicefrac{9}{5}$.

\subsection{Designing the Algorithm}
\label{sec:design_algorithm}

\quad The new routing algorithm is presented as Algorithm~\ref{alg:new_routing_algorithm}. Given a Clos network $C_{N, R}$, and a set $\mathcal{F}$ of flows, the algorithm routes the flows in two phases. Let $\nicefrac{1}{Q}$ be a flow demand, $p$ be an approximation factor, $L$ be a lower bound on the minimum congestion $OPT$, all of which specified in the next section. In Phase 1, the algorithm finds a routing for a subset of the flows that includes all flows with demand greater than $\nicefrac{1}{Q} \times OPT$ via a matching-based procedure with congestion at most $p \times L$. In Phase 2, the algorithm finds a routing for the remaining flows via a greedy procedure without increasing congestion beyond $p \times OPT$. 

\begin{algorithm}
    \caption{Given a set $\mathcal{F}$ of flows in the Clos network $C_{N, R}$, the algorithm returns a routing $r$ for $\mathcal{F}$. Variable $c( I_i M_m )$ ($c( M_m O_j )$) maintains the congestion of link $I_i M_m$ ($M_m O_j$).}
    \label{alg:new_routing_algorithm}
    \begin{algorithmic}[1]
    
        \State $\mathcal{F}_1 \coloneq \textsc{UpholdProperties}( \mathcal{F} )$ \Comment{\textit{Phase 1:}}
        \State $r( \mathcal{F}_1 ) \coloneq \textsc{LinkDisjointRouting}( \mathcal{F}_1 )$
        
        \For{$i \in [ R ]$, $m \in [ N ]$} \Comment{\textit{Phase 2:}}
            \State $c( I_i M_m ) \coloneq \sum\limits_{\mathclap{\substack{f \in \mathcal{F}_1 : i_{(f)} = i, \\ \hspace{27pt} r_{(f)} = m}}} \text{dem}(f)$; \, $c( M_m O_i ) \coloneq \sum\limits_{\mathclap{\substack{f \in \mathcal{F}_1 : j_{(f)} = i,  \\ \hspace{27pt} r_{(f)} = m}}} \text{dem}(f)$ 
        \EndFor
        \State $\mathcal{F}_2 \coloneq \mathcal{F} \setminus \mathcal{F}_1$
        \For{$j = 1, 2, \, \cdots \,, | \mathcal{F}_2 | $}
             \State $r(f_{i_j}) \coloneq \argmin_{m \in [N] } \max \{ \, c( I_{i(f_{i_j})} M_m ), c \, ( M_m O_{j(f_{i_j})} ) \}$
             \State $c( I_{i(f_{i_j})} M_{r(f_{i_j})} ) \, \, + \coloneq \text{dem}(f_{i_j})$; \, $c( M_{r(f_{i_j})}  O_{j(f_{i_j})} ) \, \, + \coloneq \text{dem}(f_{i_j})$
        \EndFor
        \State \Return r
        
    \end{algorithmic}
\end{algorithm}

\subsubsection{Phase 1} 

\quad The first phase of the algorithm is divided into two sub-phases. Phase 1.a obtains a new instance of the problem: a new Clos network, and a subset of the flows that includes all flows with demand greater than $\nicefrac{1}{Q} \times OPT$ mapped to the new network. Phase 1.b routes these flows with congestion at most $P \coloneqq p \times L$ via a decomposition into matchings of the bipartite graph corresponding to the new Clos network.  \newline

\textbf{Phase 1.a.} The new Clos network is denoted by $C_{N, R \times K}$, with $K = \lceil \nicefrac{F}{N} \rceil$ and $F$ the maximum number of flows incident to a ToR switch in the original one. The new network is obtained from the original one by creating $K$ copies of each ToR switch. The $k$'th copy of the $i$'th input switch is denoted by $I^k_i$, and the $l$'th copy of the $j$'th output switch by~$O^l_j$.

The subset of the flows mapped to the new instance is denoted by $\mathcal{F}_1$. Consider that~flows in $\mathcal{F}$ are indexed in decreasing order of demands, $\mathcal{F} \coloneq \{f_1, f_2, \dots, f_{|\mathcal{F}|} \}$ with $i \le j$ implying $\text{dem}( f_i ) \ge \text{dem}( f_j )$. The subset $\mathcal{F}_1$ is obtained from $\mathcal{F}$ and mapped to $C_{N, R \times K}$ as follows:

\begin{enumerate}[label=(\textbf{\arabic*)}]
    \item The incidence of flows in ToR switches in the original network is respected in the new network, that is, if a flow in the original network leaves an input switch $I_i$ and enters an output switch $O_j$, then in the new network it leaves a copy of $I_i$ and enters a copy of $O_j$. 
    \item The assignment from flows to copies of input-output switch pairs in the new network is such that $\mathcal{F}_1$ is a maximal subset of $\mathcal{F}$ satisfying the properties P1 through P3; in particular, it is the maximal such subset returned by the procedure \textsc{UpholdProperties} (line 1 of Algorithm~\ref{alg:new_routing_algorithm}). The properties and the procedure are introduced next. 
\end{enumerate}

\textit{Properties.} Given a flow $f \in \mathcal{F}_1$, let $k(f)$ be the copy of its input switch that $f$ leaves, and let $l(f)$ be the copy of its output switch that $f$ enters. Denote by $^{+}N^k_i$ the number of flows in $\mathcal{F}_1$ leaving $I^k_i$, and by $^{+}D^k_i$ the maximum demand over all flows in $\mathcal{F}_1$ leaving $I^k_i$, 
\begin{align*}
    ^{+}N^k_i \coloneq | \{ f \in \mathcal{F}_1 \, | \, i(f) = i \text{ and } k(f) = k \} | \,  
    \text{ and } \,
    ^{+}D^k_i \coloneq \max_{ \substack{ f \in \mathcal{F}_1 : i(f) = i \\ \hspace{27pt} k(f) = k } } \text{dem}( f ).
\end{align*}
Likewise, denote by $^{-}N^k_i$ the number of flows in $\mathcal{F}_1$ entering $O^k_i$, and by $^{-}D^k_i$ the maximum demand over all flows in $\mathcal{F}_1$ entering $O^k_i$. 

\begin{enumerate}[label=(\textbf{P\arabic*)}]
    \item For all $i \in [ R ]$ and $k \in [K]$, $^{+} N^k_i \le N$ and $^{+} N^{k}_i > 0$ with $k > 1$ implies  $^{+} N^{k-1}_i = N$. Similarly, for all $i \in [ R ]$ and $k \in [K]$, $^{-} N^k_i \le N$ and $^{-} N^{k}_i > 0$ with $k > 1$ implies  $^{-} N^{k-1}_i = N$. In words, the number of flows per copy of a ToR switch is at most $N$, with flows incident to a common switch assigned to copies of that switch from the lowest to the highest copy, such that a flow is assigned to a higher copy if there are $N$ flows incident to each of the lower copies.
    \item For all $f_i, f_j \in \mathcal{F}_1$ such that $i(f_i) = i(f_j)$ and $i \le j$, $k(f_i) \le k(f_j)$. Similarly, for all $f_i, f_j \in \mathcal{F}_1$ such that $j(f_i) = j(f_j)$ and $i \le j$, $l(f_i) \le l(f_j)$. In words, the assignment of flows incident to a common ToR switch to copies of that switch is in non-increasing order of demands.
    \item For all $i \in [ R ]$, $^{+} N^{Q}_i > 0$ implies $\sum_{k \in [ K ] } {}^{+}D^k_i \le P$. Similarly, for all $i \in [ R ]$,  $^{-} N^{Q}_i > 0$ implies $\sum_{k \in [ K ] } {}^{-} D^l_k \le P$. In words, if there are $N$ flows incident to each of the lowest $Q-1$ copies of a ToR switch, then the total demand over the highest demand flows incident to each of the copies of that switch is at most $P$.
\end{enumerate}

The previous properties are illustrated in Figure~\ref{fig:new_routing_algorithm}. \newline 

\begin{figure}[t!]
    \centering
    \captionsetup{labelfont=bf}
    \begin{subfigure}{.49\textwidth}
        \centering
        \captionsetup{labelfont=bf}    
        \includegraphics[width=0.85\textwidth]{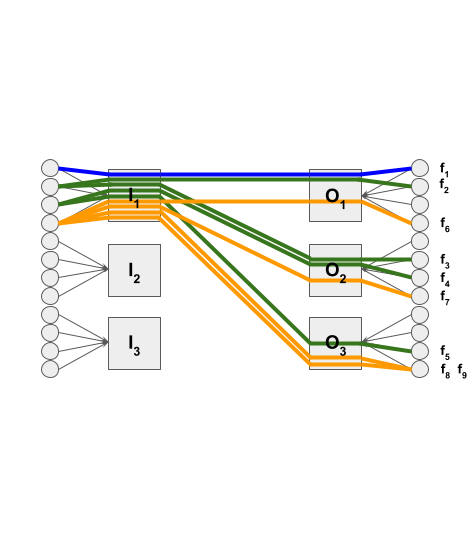} 
        \subcaption{A set of flows in the original network.}
    \end{subfigure}
    \begin{subfigure}{.49\textwidth}
        \centering
        \captionsetup{labelfont=bf}
        \includegraphics[width=0.85\textwidth]{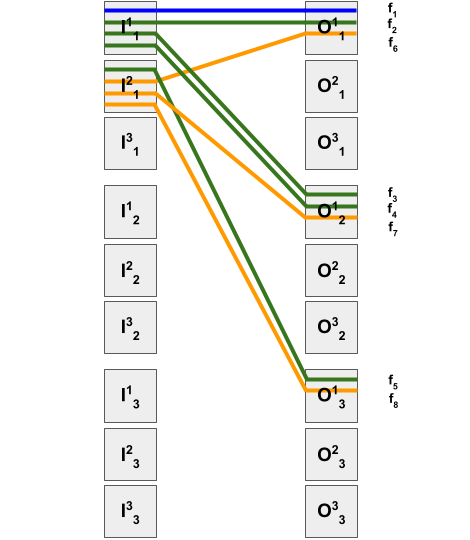} 
        \subcaption{A subset of the flows mapped to the new network.}
    \end{subfigure}

    \vspace{10pt}

    \begin{subfigure}{.49\textwidth}
        \centering
        \captionsetup{labelfont=bf}
        \includegraphics[width=0.85\textwidth]{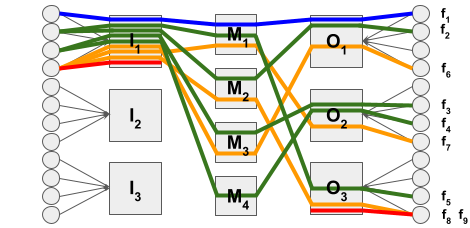} 
        \subcaption{Routing of the flows upon Phase 1.}
    \end{subfigure}
    \begin{subfigure}{.49\textwidth}
        \centering
        \captionsetup{labelfont=bf}    
        \includegraphics[width=0.85\textwidth]{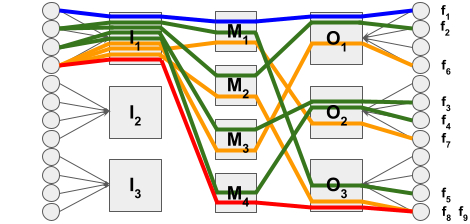} 
        \subcaption{Routing of the flows upon Phase 2.}
    \end{subfigure}
    
    \caption{Algorithm~\ref{alg:new_routing_algorithm} with $Q = 3$ and $P = \nicefrac{5}{3}$ for a set of flows in the Clos network $C_{4, 3}$. Flow $f_1$ has demand $1$ (in blue), flows $f_2$ through $f_5$ have demand $\nicefrac{1}{2}$ (in green), and flows $f_6$ through $f_9$ have demand $\nicefrac{1}{4}$ (in orange). While the subset of the flows in the new instance satisfy properties $P1$ through $P3$, the complete set does not: if $f_9$ (in red) is included in the subset, then it must be assigned to $(I^3_1, O^1_3)$, in which case, since $^{+} N^{3}_i > 0$, P3 implies $1 + \nicefrac{1}{2} + \nicefrac{1}{4} < \nicefrac{5}{3}$. In the end of Phase 1, the minimum congestion path from $I_1$ to $O_3$ is $M_4$, hence, in Phase 2, flow $f_9$ is assigned to $M_4$.}
    \label{fig:new_routing_algorithm}
\end{figure}

\textit{Procedure.} Given the set $\mathcal{F}$ of flows in $C_{N, R}$, the procedure $\textsc{UpholdProperties}( \mathcal{F} )$ returns the posited subset $\mathcal{F}_1$. Initially, the set $\mathcal{F}_1$ is empty. The procedure tests each flow in $\mathcal{F}$ for inclusion in $\mathcal{F}_1$ in decreasing order of demands. For each $i \in [ \, | \mathcal{F} | \, ]$, the procedure includes $f_i$ in $\mathcal{F}_1$ if there are copies $k$ and $l$ such that the assignment of $f_i$ to $(I^k_{i(f_i)}, O^l_{j(f_i)})$ satisfies properties P1 through P3. 

In detail, denote by $\mathcal{G} \subseteq \mathcal{F}_1$ the set of flows included in $\mathcal{F}_1$ prior to testing $f_i$. Let $x$ be the lowest copy of $I_{i(f_i)}$ such that the number of flows in $\mathcal{G}$ leaving that copy is strictly less than $N$, $x \coloneq \min \{ k \in [ K ] \, | \, ^{+}N^k_{i(f_i)} < N \}$. Likewise, let $y \coloneq \min \{ l \in [ K ] \, | \, ^{-}N^l_{j( f_i )} < N \}$. We say that $I_{i(f_i)}$ \textit{accepts} $f_i$ if it holds that: 
\begin{align*}
    x \ge Q \, \text{ implies } \sum_{k \in [ x - 1 ] } {}^{+} D^k_{i(f_i)} + \max \{ ^{+} D^{x}_{i(f_i)}, \text{dem}( f_i ) \} \le P; 
\end{align*}
and that $I_{i(f_i)}$ \textit{rejects} $f_i$ otherwise. Since flows are tested in decreasing order of demands, if $^{+} N^{x}_{i(f_i)} > 0$, then $I_{i(f_i)}$ accepts $f_i$. Likewise, we say that $O_{j(f_i)}$ \textit{accepts} $f_i$ if it holds that: 
\begin{align*}
    y \ge Q \, \text{ implies } \sum_{l \in [ y - 1 ] } {}^{-} D^l_{j(f_i)} + \max \{ ^{-} D^{y}_{j(f_i)}, \text{dem}( f_i ) \} \le P; 
\end{align*}
and that $O_{j(f_i)}$ \textit{rejects} $f_i$ otherwise. Flow $f_i$ is included in $\mathcal{F}_1$ if its accepted by both $I_{i(f_i)}$ and $O_{j(f_i)}$, in which case it is assigned to $(I^{x}_{i(f_i)}, O^{y}_{j(f_i)})$. \newline

\textbf{Phase 1.b.} The flows in $\mathcal{F}_1$ are routed in the original network according to a link-disjoint routing in the new network (line 2 of Algorithm~\ref{alg:new_routing_algorithm}). Since the number of flows incident to a copy of a ToR switch is at most the number $N$ of middle switches, there is one such routing in the new network, and it can be found in polynomial time (see Theorem~\ref{thm:rearrangeability}). 

The routing of $\mathcal{F}_1$ assigns the flows incident to a common copy of a ToR switch to different middle switches, with at most $K$ flows incident to that ToR switch assigned to the same middle switch. Therefore, the maximum congestion over all links incident to a ToR switch is at most the total demand over the highest demand flows incident to each copy of that switch. If there are $N$ flows incident to each of the lowest $Q-1$ copies of a ToR switch, then this congestion is at most $P$. In the analysis of the next section, $Q$ is set such that the total demand over the highest demand flows incident to the first $Q-1$ copies of a ToR switch is also at most $P$, so that the congestion of the routing of $\mathcal{F}_1$ is at most $P$.

\subsubsection{Phase 2} 

\quad The second phase of the algorithm routes each flow not routed in Phase 1. The subset of the flows not routed in Phase 1 is denoted by $\mathcal{F}_2 \coloneq \mathcal{F} \setminus \mathcal{F}_1$. Consider that the flows in $\mathcal{F}_2$ are indexed in decreasing order of demands, $\mathcal{F}_2 \coloneq \{ f_{i_1}, f_{i_2}, \dots f_{i_{|\mathcal{F}_2|}} \}$ with $i_j \in [ \, | \mathcal{F} | \, ]$ and $i_{j} \le i_{j+1}$ for all $j \in [ \, | \mathcal{F}_2 | - 1 ]$. Initially, the congestion of each link is the total demand over all flows in $\mathcal{F}_1$ routed through that link (line 4). The algorithm routes each flow in $\mathcal{F}_2$ in decreasing order of demands. For each  $j \in [ \, | \mathcal{F}_2 | \, ]$, the algorithm routes $f_{i_j}$ on a path whose congestion is minimum, breaking ties arbitrarily (line~8). Then, the algorithm updates the congestion of that path (line~9). The execution of Algorithm~\ref{alg:new_routing_algorithm} is illustrated in Figure~\ref{fig:new_routing_algorithm}.

\subsection{Analyzing the Algorithm}
\label{sec:analysis_algorithm}

\quad Define the lower bound $L$ on the optimal congestion $OPT$ as follows:
\begin{align*}
    L & \coloneq \max_{i \in [ R ]} \max \{ {}^{+} L_i, {}^{-} L_i \},
\end{align*}
where
\begin{align*}
        {}^{+} L_i & \coloneq \max \Big\{ 
            \max_{ f \in \mathcal{F} : i(f) = i } \text{dem}( f ), 
            \nicefrac{1}{N} \times \sum_{f \in \mathcal{F} : i(f) = i} \text{dem}( f ) \Big\} \text{ and } \\    
        {}^{-} L_i & \coloneq \max \Big\{ 
            \max_{ f \in \mathcal{F} : j(f) = i } \text{dem}( f ), 
            \nicefrac{1}{N} \times \sum_{f \in \mathcal{F} : j(f) = i} \text{dem}( f ) \Big\}. 
\end{align*}
In words, $OPT$ is at least the highest demand over all flows leaving $I_i$ (entering $O_i$), and at least the total demand over all flows leaving $I_i$ (entering $O_i$) averaged over the number of middle switches; if flows were splittable rather than unsplittable, then the latter term would correspond to the congestion of each link leaving $I_i$ (entering $O_i$) in the minimum congestion routing that divides the demand of each flow evenly over all source-destination paths.

The next theorem shows that, if $P$ is set to $\nicefrac{9}{5} \times L$ and $Q$ to $3$, then the algorithm approximates a minimum congestion routing by a factor at most $\nicefrac{9}{5}$. 

\begin{theorem}
    \label{thm:new_routing_algorithm}
    Fix $p \coloneq \nicefrac{9}{5}$ and $q \coloneq 3$. For every set $\mathcal{F}$ of flows, if Algorithm~\ref{alg:new_routing_algorithm} sets $P = p \times L( \mathcal{F} )$ and $Q = q$, then it returns a routing with congestion at most $p \times OPT( \mathcal{F} )$.
\end{theorem}

\textbf{Preliminaries to the proof of Theorem~\ref{thm:new_routing_algorithm}.} The proof of the theorem makes use of the additional properties of the set $\mathcal{F}_1$ of flows routed in Phase 1 detailed in the next lemma. Denote by ${}^{+} d^k_i$ the minimum demand over all flows in $\mathcal{F}_1$ leaving $I^k_i$, 
\begin{align*}
    ^{+}d^k_i \coloneq \min_{ \substack{ f \in \mathcal{F}_1 : i(f) = i \\ \hspace{23pt} k(f) = k } } \text{dem}( f ).
\end{align*}
Likewise, denote by ${}^{-} d^k_i$ the minimum demand over all flows in $\mathcal{F}_1$ entering $O^k_i$. 

\begin{lemma}
    \label{lem:new_routing_algorithm}
    The set $\mathcal{F}_1$ of flows satisfies the following additional properties:
    
        \begin{enumerate}[label=(\textbf{Q\arabic*)}]
            \item Assume that $p \ge \sum_{k \in [q - 1]} \nicefrac{1}{k}$. For all $i \in [ R ]$, $\sum_{k \in [ K ]} {}^{+} D^k_i \le p \times OPT$. Similarly, for all $j \in [ R ]$, $\sum_{l \in [ K ]} {}^{-} D^l_j \le p \times OPT$. In words, the congestion of the routing of $\mathcal{F}_1$ is at most $p \times OPT$.
            \item For all $f \in \mathcal{F}$, \, $\text{dem}( f ) > \nicefrac{1}{q} \times OPT$ \text{ implies } $f \in \mathcal{F}_1$. In words, every flow with demand greater than $\nicefrac{1}{q} \times OPT$ is routed in Phase 1.
            \item Assume that there is $f \in \mathcal{F}_2$ such that $I_{i(f)}$ rejected $f$ when it was tested in Phase 1. Then, $\sum_{k \in [ K ]} {}^{+} d^k_{i(f)} > (p-1) \times L$. Similarly, assume that there is $f \in \mathcal{F}_2$ such that $O_{j(f)}$ rejected $f$ when it was tested in Phase 1. Then, $\sum_{l \in [ K ]} {}^{-}  d^l_{j(f)} > (p-1) \times L$. In words, if there is a flow not routed in Phase 1, then, assuming that it was rejected by its input switch, every link incident to that switch has congestion greater than $(p-1) \times L$. 
        \end{enumerate}
\end{lemma}

The proof of the lemma requires two preparatory claims. Since the proofs of the statements in the lemma and in the preparatory claims are analogous for both input and output switches, they are given only for input switches.

\begin{claim}
    \label{claim:new_routing_algorithm_1}
    For all $i \in [ R ]$ and all positive integers $k$, the number of flows $f \in \mathcal{F}$ such that $i( f ) = i$ and $\text{dem}( f ) > \nicefrac{1}{k} \times OPT$ is at most $N \times (k - 1)$. Similarly, for all $j \in [ R ]$ and all positive integers $l$, the number of flows $f \in \mathcal{F}$ such that $j( f ) = j$ and $\text{dem}( f ) > \nicefrac{1}{l} \times OPT$ is at most $N \times (l - 1)$.
\end{claim}

\begin{proof}
    Assume, by contradiction, that the number of flows $f \in \mathcal{F}$ such that $i( f ) = i$ and $\text{dem}( f ) > \nicefrac{1}{k} \times OPT$ is at least $N \times (k - 1) + 1$. Then, for every routing of these flows, there is a middle switch to which at least $k$ flows are assigned. Therefore, since the demand of each of these $k$ flows is strictly greater than $\nicefrac{1}{k} \times OPT$, we conclude that the congestion a minimum congestion routing for $\mathcal{F}$ is strictly greater than $OPT$, thus arriving at a contradiction. 
\end{proof}

\begin{claim}
    \label{claim:new_routing_algorithm_2}
    For all $f \in \mathcal{F}_1$, $\text{dem}( f ) \le \nicefrac{1}{k(f)} \times OPT$. Similarly, for all $f \in \mathcal{F}_1$, $\text{dem}( f ) \le \nicefrac{1}{l(f)} \times OPT$. 
\end{claim}

\begin{proof}
    Assume, by contradiction, that there is $f \in \mathcal{F}_1$ such that $\text{dem}( f ) > \nicefrac{1}{k(f)} \times OPT$. From P1, we know that $^{+} N^k_i = N$ for all $k \in [ k(f) - 1 ]$, and, from P2, that, for all flow $g \in \mathcal{F}_1$ such that $i(g) = i(f)$ and $k( g ) < k( f )$, $\text{dem}( f ) > \nicefrac{1}{k(f)} \times OPT$. Therefore, we conclude that the number of flows $g \in \mathcal{F}$ such that $i( g ) = i( f )$ and $\text{dem}( g ) > \nicefrac{1}{k(f)} \times OPT$ is at least $N \times (k(f) - 1) + 1$. However, from Claim~\ref{claim:new_routing_algorithm_1}, we know that this number is at most $N \times (k(f) - 1)$, thus arriving at a contradiction.  
\end{proof}

\begin{proof}[\textbf{Proof of Lemma~\ref{lem:new_routing_algorithm}}]
    We show that each of the properties Q1 through Q3 is satisfied. \newline
    
    \textit{Property Q1:} From P3, we know that ${}^{+} N^q_i = 0$ or $\sum_{k \in [ K ]} {}^+ D^k_i \le p \times L$. If ${}^+ N^q_i = 0$, then, from Claim~\ref{claim:new_routing_algorithm_2}, we write:
    \begin{align*}
        \sum_{k \in [ K ]} {}^+ D^k_i & =   \sum_{k \in [ q - 1 ]} {}^+ D^k_i \\
                                      & \le \sum_{k \in [ q - 1 ]} \nicefrac{1}{k} \times OPT && (\text{from Claim}~\ref{claim:new_routing_algorithm_2})\\
                                      & \le p \times OPT && (\text{from the assumption on $p$}).
    \end{align*}
    If ${}^+ N^q_i \neq 0$, then we write:
    \begin{align*}
        \sum_{k \in [ K ]} {}^+ D^k_i & \le p \times L \\
                                      & \le p \times OPT.
    \end{align*}
    In both cases, we conclude that $\sum_{k \in [ K ]} {}^{+} D^k_i \le p \times OPT$. \newline

    \textit{Property Q2:} Assume, by contradiction, that $\text{dem}( f ) > \nicefrac{1}{q} \times OPT$ but $f \in \mathcal{F}_2$. Denote by $\mathcal{G} \subseteq \mathcal{F}_1$ the set of flows included in $\mathcal{F}_1$ prior to testing $f$. From P1 and P3, ${}^{+} N^k_{i(f)} = N$ for all $k \in [ q - 1 ]$, or ${}^{-}  N^l_{j(f)} = N$ for all $l \in [ q - 1 ]$. Assume, without loss of generality, that ${}^{+} N^k_{i(f)} = N$ for all $k \in [ q - 1 ]$. Since flows are tested in non-increasing order of demands, we conclude that the number of flows $g \in \mathcal{G}$ such that $i( g ) = i( f )$ and $\text{dem}( g ) > \nicefrac{1}{q} \times OPT$ is at least $N \times (q - 1) + 1$. However, from Claim~\ref{claim:new_routing_algorithm_1}, we know that this number is at most $N \times (q - 1)$, thus arriving at a contradiction. \newline
    
    \textit{Property Q3:} Denote by $\mathcal{G} \subseteq \mathcal{F}_1$ be the set of flows included in $\mathcal{F}_1$ prior to testing $f$, and by $x$ the lowest copy of $I_{i(f)}$ such that the number of flows in $\mathcal{G}$ leaving that copy is strictly less than $N$. From P2 and P3, we write:
    
    \begin{align*}
         \sum_{k \in [ K ]} {}^{+} d^k_{i(f)}
                                       & \ge \sum_{k \in [ x-2 ]} {}^{+} d^{k}_{i(f)} + {}^{+} d^{x-1}_{{i(f)}} \\
                                       & \ge \sum_{k \in [ 2, x-1 ]} {}^{+} D^{k}_{i(f)} + \max \{ {}^{+} D^{ x }_{i(f)}, \text{dem}( f ) \} && (\text{from P2}) \\
                                       & = \underbrace{\sum_{k \in [ x-1 ]} {}^{+} D^{k}_{i(f)} + \max \{ {}^{+} D^{ x}_{i(f)}, \text{dem}( f ) \}}_{> \, p \times L, \text{ since $I_{i(f)}$ rejected $f$}} - {}^{+} D^{1}_{i(f)} \\ 
                                       & > (p-1) \times L, && (\text{from P3})
    \end{align*}
    to conclude that $\sum_{k \in [ K ]} {}^{+} d^k_{i(f)} > (p-1) \times L$.
\end{proof}

\begin{proof}[\textbf{Proof of Theorem~\ref{thm:new_routing_algorithm}}]
    The proof of the theorem is by contradiction. By choice of $p$ and $q$, $p \ge \sum_{ k \in [q-1] } \nicefrac{1}{k}$. Consequently, from Q1, we know that the congestion of the routing for $\mathcal{F}_1$ is at most $p \times OPT$. Consider the iteration in Phase 2 when a flow $f \in \mathcal{F}_2$ is routed. Assume, without loss of generality, that $I_{i(f)}$ rejected $f$ when it was tested in Phase 1. Denote by $\mathcal{H}$ the set of flows routed in Phase 1 or in previous iterations in Phase 2. We say that a link $I_i M_m$ ($M_m O_j$) is \textit{congested} if the total demand over all flows in $\mathcal{H}$ routed through $I_i M_m$ ($M_m O_j$) is greater than $D \coloneq p \times OPT - \text{dem}( f )$,
    \begin{align*}
        \sum_{ \substack{ h \in H : i(h) = i(f), \\ \text{ and } r(h) = r(f)} } \text{dem}( h ) > D 
        \hspace{15pt}
        \big(  \sum_{ \substack{ h \in H : i(h) = i(f), \\ \text{ and } r(h) = r(f)} } \text{dem}( h ) > D \, \big),
    \end{align*}
   and that a path $I_i M_m O_j$ is congested if both $I_i M_m$ and $M_m O_j$ are congested. Assume, by contradiction, that $I_{i(f)} M_m O_{j(f)}$ is congested for all $m \in [ N ]$. 
   
   The derivation of the contradiction makes use of the next two claims. Let $D$ be the congestion of a congested link, and let $C$ be the number of congested links leaving $I_{i(f)}$. 

    \vspace{3pt}

    \textit{Claim (i): $D$ satisfies $D \ge (p - \nicefrac{1}{q}) \times OPT$}. From Q2, we write 
    \begin{align*}
        D & \triangleq p \times OPT - \text{dem}( f )                                     \\
                 & \ge p \times OPT - \nicefrac{1}{q} \times OPT                           \\
                 & =  (p - \nicefrac{1}{q}) \times OPT.  \hspace{278pt} \vartriangleleft
    \end{align*}

    \textit{Claim (ii): $C$ satisfies $C > N \times (1 - \frac{1}{p-\nicefrac{1}{q}})$}. From the assumption that $I_{i(f)} M_m O_{j(f)}$ is congested for all $m \in [ N ]$, we know that if $I_{i(f)} M_m$ is not congested, then $M_m O_{i(f)}$ is congested, and we deduce that
    \begin{align*}
        \sum_{h \in \mathcal{H} : j(h) = j(f)} \text{dem}( h ) > (N - C) \times D \Leftrightarrow 
        C > N - \frac{1}{D} \sum_{h \in \mathcal{H} : j(h) = j(f)} \text{dem}( h ).
    \end{align*}
    Then, from Claim (i), we write
    \begin{align*}
        C & >   N - \frac{1}{D} \times \sum_{h \in \mathcal{H} : j(h) = j(f)} \text{dem}( h ) \\ 
          & \ge N - \frac{1}{( p - \nicefrac{1}{q} ) \times OPT} \times  \sum_{h \in \mathcal{H} : j(h) = j(f)} \text{dem}( h ) \\ 
          & \ge N \times (1 - \frac{1}{p-\nicefrac{1}{q}}). \hspace{273pt} \vartriangleleft
    \end{align*}

    The contradiction follows from the forthcoming lower bound on the total demand over all flows in $\mathcal{H}$ leaving $I_{i(f)}$. From Q3 and the previous claims, we write
    \begin{align*}
        \sum_{h \in \mathcal{H} : i(h) = i(f)} \text{dem}( h ) & > C \times D + (N - C) \times (p-1) \times L \\
        & = N \times L \times (p - 1) + C \times (D - (p - 1 ) \times L )    \\
        & \ge N \times L \times (p - 1) + C \times (D - (p - 1 ) \times OPT )  \\
        & > N \times L \times (p - 1) + N \times OPT \times (1 - \frac{1}{p-\nicefrac{1}{q}})(1 - \nicefrac{1}{q}) \\
        & \ge \sum_{h \in \mathcal{H} : i(h) = i(f)} \text{dem}( h ) \times (p - 1 + (1 - \frac{1}{p-\nicefrac{1}{q}})(1 - \nicefrac{1}{q}) ).
    \end{align*}
    By setting $p \triangleq \nicefrac{9}{5}$ and $q \triangleq \nicefrac{1}{3}$, we write
    \begin{align*}
        p - 1 + \Big(1 - \frac{1}{p-\nicefrac{1}{q}} \Big) \Big(1 - \nicefrac{1}{q}\Big) = \nicefrac{167}{165}
    \end{align*}
    to conclude that
    \begin{align*}
        \sum_{h \in \mathcal{H} : i(h) = i(f)} \text{dem}( h ) > \sum_{h \in \mathcal{H} : i(h) = i(f)} \text{dem}( h ). \hspace{274pt} \qedhere
    \end{align*}
\end{proof}

The proof of the previous theorem can be generalized to show that, if $P$ and $Q$ are set as before, then not only does the algorithm approximate the minimum congestion by a factor at most $\nicefrac{9}{5}$, but it also returns a routing with congestion at most $\nicefrac{9}{5}$. This generalization only involves introducing the preparatory Claim~\ref{claim:new_routing_algorithm_3} analogous to Claim~\ref{claim:new_routing_algorithm_1}, and replacing the term $OPT$ with the term $\min \{ OPT, 1 \}$ everywhere in the analysis. 

\begin{claim}
    \label{claim:new_routing_algorithm_3}
    For all $i \in [ R ]$ and all positive integers $k$, the number of flows $f \in \mathcal{F}$ such that $i( f ) = i$ and $\text{dem}( f ) > \nicefrac{1}{k}$ is at most $N \times (k - 1)$. Similarly, for all $j \in [ R ]$ and all positive integers $l$, the number of flows $f \in \mathcal{F}$ such that $j( f ) = j$ and $\text{dem}( f ) > \nicefrac{1}{l}$ is at most $N \times (l - 1)$.
\end{claim}

\begin{proof}
    From the assumption that the total demand over all flows leaving a source is at most $1$, we know that, for all $s \in [ N ]$, the number of flows $f \in \mathcal{F}$ such that $i(f) = f$, $s(f) = s$, and $\text{dem}( f ) > \nicefrac{1}{k}$ is at most $k-1$. Therefore, the number of flows $f \in \mathcal{F}$ such that $i(f) = f$ and $\text{dem}( f ) > \nicefrac{1}{k}$ is at most $N\times (k-1)$.
\end{proof}

\begin{theorem}
    Fix $p \coloneq \nicefrac{9}{5}$ and $q \coloneq 3$. For every set $\mathcal{F}$ of flows, if Algorithm~\ref{alg:new_routing_algorithm} sets $P = p \times L$ and $Q = q$, then it returns a routing with congestion at most $p \times \min \{ OPT( \mathcal{F} ), 1 \}$.
\end{theorem}

%% file: 5-offline_lower_bounds.tex
\quad We present lower bounds on worst-case congestion and approximation of minimum congestion routings in Clos networks by polynomial-time algorithms. In \S\ref{sec:limits_congestion}, we show that there are sets of flows for which the congestion of a minimum congestion routing is at least $\nicefrac{3}{2}$, and in \S\ref{sec:limits_approximation}, we show that it is impossible for any polynomial-time algorithm to approximate a minimum congestion routing by a factor less than $\nicefrac{3}{2}$ unless $P = NP$. \newline

\textbf{The cross gadget.} The building block of the constructions yielding the posited lower bounds is the following set of flows. The \textit{cross gadget} of size $N$, for some integer $N \ge 2$, illustrated in Figure~\ref{fig:gadget_macro} for $N = 3$, corresponds to the set of flows in the Clos network $C_{N, N}$ where for each of the $N$ input switches there are $N-1$ flows with demand $1$ leaving that input switch and entering each of the first $N-1$ output switches:

\begin{itemize}
    \item There is one flow $(s^{j}_i, t^{i}_j)$ with demand $1$, for all $i \in [ N ]$ and $j \in [ N-1 ]$.
\end{itemize}

The key property of the cross gadget, detailed in the next lemma, is that there is a unique routing with congestion $1$ (modulus the numbering of the middle switches), which assigns the flows incident to a common ToR switch to different middle switches, and does not assign a flow to a different middle switch at each input switch. 

\begin{lemma}
    \label{lem:gadget}
    Consider the cross gadget of size $N$. For every routing with congestion $1$, the following properties hold:
    \begin{enumerate}[label=(\textbf{\arabic*)}]
        \item For all $i \in [ N ]$, the $N-1$ flows leaving $I_i$ are assigned to $N-1$ different middle switches, and, for all $j \in [ N-1 ]$, the $N$ flows entering $O_j$ are assigned to $N$ different middle switches.
        \item For all $i_1, i_2 \in [ N ]$ such that $i_1 \neq i_2$, the middle switch to which no flow leaving $I_{i_1}$ is assigned differs from the middle switch to which no flow leaving $I_{i_2}$ is assigned.         
    \end{enumerate}
\end{lemma}

\begin{proof}
    The first property follows from the fact that all flows have demand $1$. We show that~the second property holds. First, since the $N$ flows entering each of the first $N{-}1$ output switches~are assigned to $N$ different middle switches, we deduce that each middle switch is assigned $N{-}1$ flows. Second, since the $N{-}1$ flows leaving each of the $N$ input switches are assigned to $N{-}1$ different middle switches, we further deduce that for each middle switch the $N{-}1$ flows assigned to it leave $N{-}1$ different input switches. Therefore, if there are two input switches $I_{i_1}$ and $I_{i_2}$, $i_1 \neq i_2$, and a middle switch $M_{m}$ such that no flow leaving $I_{i_1}$ is assigned to $M_m$ and no flow leaving $I_{i_2}$ is assigned to $M_m$, then there are at most $N-2$ flows assigned to $M_m$, which contradicts the existence of $N-1$ flows assigned to $M_m$.
\end{proof}

\begin{figure}
    \centering
    \captionsetup{labelfont=bf}
    \begin{subfigure}{.49\textwidth}
        \centering
        \captionsetup{labelfont=bf}
        \includegraphics[width=0.8\textwidth]{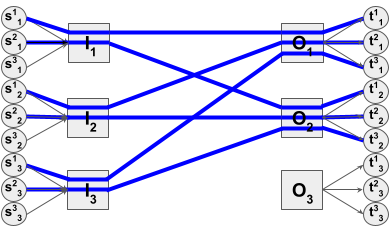} 
        \subcaption{The cross gadget without routing.}
        \label{fig:gadget_macro}
    \end{subfigure}
    \begin{subfigure}{.49\textwidth}
        \centering
        \captionsetup{labelfont=bf}
        \includegraphics[width=0.8\textwidth]{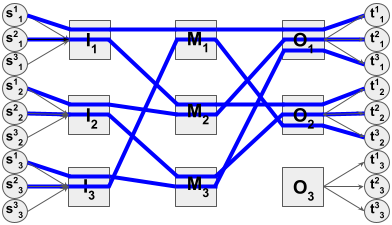} 
        \subcaption{Elemental routing of the cross gadget.}
        \label{fig:gadget_Clos}
    \end{subfigure}
    \caption{The cross gadget of size ${N = 3}$.}
    \label{fig:gadget}
\end{figure}

Therefore, all routings of the cross gadget with congestion $1$ reduce to the routing described below and illustrated in Figure~\ref{fig:gadget_Clos} for $N = 3$.

\begin{itemize}
    \item The flow $(s^{j}_i, t^{i}_j)$ is assigned to $M_{m+1}$, where $m = i + j - 2 \pmod N$, for all $i \in [ N ]$ and $j \in [ N-1 ]$.
\end{itemize}

\noindent We call this routing the \textit{elemental} routing for the cross gadget.

\subsection{Limits to Congestion}
\label{sec:limits_congestion}

\quad \enskip If all flows have demand $1$, in which case there is at most one flow per source and per destination, then for every set of flows there is a routing with congestion $1$. (This is a corollary of Theorem~\ref{thm:rearrangeability}.) In contrast, the forthcoming theorem shows that if this premise is relaxed such that each flow has demand $1$ or $\nicefrac{1}{2}$, in which case there are at most two flows per source and per destination, then there are now sets of flows for which every routing has congestion greater than $1$. 

\begin{theorem}
    \label{thm:limit_congestion}
    Consider a Clos network $C_{N, R}$, for $N \ge 2$ and $R \ge N + 1$. There is a set of flows with demands $1$ or $\nicefrac{1}{2}$ for which the congestion of a minimum congestion routing is $\nicefrac{3}{2}$. 
\end{theorem}

\begin{proof}[\textbf{Proof of Theorem~\ref{thm:limit_congestion}}]
    We design a set of flows in the Clos network $C_{N, N+1}$ composed of flows with demand $1$ or $\nicefrac{1}{2}$ for which every routing has congestion at least $\nicefrac{3}{2}$. Since increasing the number of input and output switches from $N+1$ to $R$ does not change the number of source-destination paths from the first $N+1$ input switches to the first $N+1$ output switches, if a set of flows satisfies the theorem requirements in the network $C_{N, N+1}$, then the same set also satisfies them in the network $C_{N, R}$. Consequently, the conclusion follows. \newline

    \textbf{Designing the set of flows:} The posited set of flows, illustrated in Figure~\ref{fig:limit_congestion_1} for $N = 3$, is composed of three types of flows:
    
    \begin{itemize}
        \item The type 1 flows (in blue) correspond to the cross gadget of size $N$.
        \item \textit{Type 2 (in orange):} There is one flow $(s^{N}_i, t^{\lceil \nicefrac{i}{2} \rceil}_N)$ with demand $\nicefrac{1}{2}$, for all $i \in [ N ]$.     
        \item \textit{Type 3 (in green):} There is one flow $(s^{N}_{N+1}, t^{N}_N)$ with demand $1$. 
    \end{itemize}

    \begin{figure}
        \centering
        \captionsetup{labelfont=bf}
        \begin{subfigure}{.49\textwidth}
            \centering
            \captionsetup{labelfont=bf}
            \includegraphics[width=0.8\textwidth]{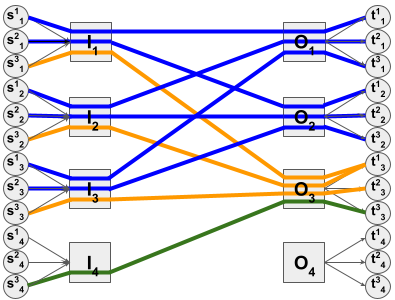} 
            \subcaption{The flows without routing.}
            \label{fig:limit_congestion_1}
        \end{subfigure}
        \begin{subfigure}{.49\textwidth}
            \centering
            \captionsetup{labelfont=bf}
            \includegraphics[width=0.8\textwidth]{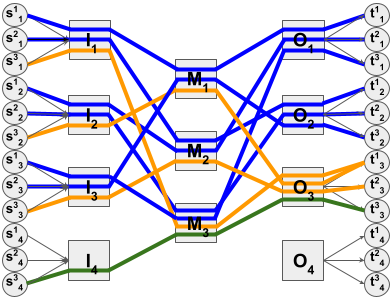} 
            \subcaption{A minimum congestion routing for the flows.}
            \label{fig:limit_congestion_2}
        \end{subfigure}
        \caption{The set of flows underlying the proof of Theorem~\ref{thm:limit_congestion} for ${N = 3}$ and $R = 4$.}
        \label{fig:limit_congestion}
    \end{figure}

    \textbf{Calculating a minimum congestion routing:} We show that the congestion of a minimum congestion routing is $\nicefrac{3}{2}$. The proof involves two steps. The first step shows that there exists a routing with congestion $\nicefrac{3}{2}$. The posited minimum congestion routing is described below and illustrated in Figure~\ref{fig:limit_congestion_2} for $N = 3$.
    
    \begin{itemize}
        \item The type 1 flows are assigned according to the elemental routing for the cross gadget.
        \item The type 2 flow $(s^{N}_i, t^{\lceil \nicefrac{i}{2} \rceil}_N)$ is assigned to $M_{m+1}$, where $m = i - 2 \pmod N$, for all $i \in [ N ]$.
        \item The type 3 flow is assigned to $M_N$.
    \end{itemize}
    Every edge except for $M_N O_N$ is traversed by at most one flow, and thus has congestion at most $1$. Edge $M_N O_N$ is traversed by both the type 2 flow leaving $I_1$ and the type 3 flow, and thus has congestion $\nicefrac{3}{2}$. Therefore, the routing has congestion $\nicefrac{3}{2}$.
    
    The second step shows that there is no routing with congestion less than $\nicefrac{3}{2}$. First, from Lemma~\ref{lem:gadget}, we know that every routing of the type 1 flows with congestion less than $\nicefrac{3}{2}$ reduces to the elemental routing for the cross gadget, and thus assume that the type 1 flows are routed according to it. Then, we distinguish two cases, which assert that every subsequent routing of the type 2 and the type 3 flows lead to congestion at least $\nicefrac{3}{2}$, thereby concluding the proof.  
    
    \begin{itemize}
        \item \textit{Case 1:} If for all $i \in [ N ]$ the type 2 flow leaving $I_i$ is assigned to the middle switch to which no type 1 flows leaving $I_i$ is assigned to, then the $N$ type 2 flows are assigned to $N$ different middle switches (as in Figure~\ref{fig:limit_congestion_2}). Consequently, the type 3 flow is assigned to the same middle switch than some type 2 flow, in which case the routing has congestion~$\nicefrac{3}{2}$.   
        \item \textit{Case 2:} Otherwise, there is $i \in [N]$ such that the type 2 flow leaving $I_i$ is assigned to the same middle switch than some type 1 flow leaving $I_i$, in which case the routing has congestion at least~$\nicefrac{3}{2}$. \qedhere
    \end{itemize}
\end{proof}

\subsection{Limits to Approximation}
\label{sec:limits_approximation}

\quad \enskip If all flows have demand $1$, then for every set of flows a routing with congestion $1$ can~be found in polynomial-time. (This is again a corollary of Theorem~\ref{thm:rearrangeability}.)  In contrast, the theorem below shows that if this premise is again relaxed such that each flow has demand $1$ or $\nicefrac{1}{2}$, then it becomes impossible to distinguish in polynomial-time between a routing with congestion at most $1$ and a routing with congestion at least $\nicefrac{3}{2}$. Consequently, the theorem implies that it is impossible to approximate a minimum congestion routing by a factor less than $\nicefrac{3}{2}$.

\begin{theorem}
    \label{thm:limit_approximation}
    For special case of the minimum congestion routing problem in Clos networks where flows have demand $1$ or $\nicefrac{1}{2}$, the question of deciding if there is a routing with congestion at most $1$ is $NP$-complete.
\end{theorem}

\begin{proof}[\textbf{Proof of Theorem~\ref{thm:limit_approximation}}]
    We introduce a reduction from the $3$-edge coloring problem, which is known to be $NP$-complete~\cite{Karp_75}, to the question of deciding if there~is~a routing with congestion at most $1$ for flows whose demands are $1$ or $\nicefrac{1}{2}$. In the $3$-edge coloring problem, the input is an undirected graph whose maximum vertex degree is $3$. The question is to decide if there exists a \textit{proper edge-coloring}, meaning an assignment from edges to colors such that no two adjacent edges are assigned the same color, using at most $3$ colors.  \newline

    \textbf{Description of the reduction:} Consider an input graph $G = (V(G), E(G))$ to the $3$-edge coloring problem, with vertices indexed from $1$ to $|V(G)|$, $V(G) = \{ v_1, v_2, \dots, v_{|V(G)|} \}$, and edges indexed from $1$ to $|E(G)|$, $E(G) = \{ e_1, e_2, \dots, e_{|E(G)|} \}$. For each vertex $v_k \in V(G)$, let $\text{rank}_{v_k}$ be an arbitrary ranking of the neighbors of $v_k$, such that $\text{rank}_{v_k}( v_l ) \in \{ 1, 2, 3 \}$ denotes the ranking of a neighbor $v_l$ of $v_k$. We construct an instance of the minimum congestion routing problem from $G$. The construction is illustrated in Figure~\ref{fig:limits_approximation} for a particular input graph. 
    
    The Clos network is denoted by $C_{3, 3|V(G)| + |E(G)|}$. The ToR switches are divided into $|V(G)|$ vertex blocks and $|E(G)|$ edge blocks, with the $k$'th vertex block corresponding to input switches $I_{3(k-1)+1}$ through $I_{3(k-1)+3}$ and output switches $O_{3(k-1)+1}$ through $O_{3(k-1)+3}$, $k \in [ \, | V(G) | \, ]$, and the $m$'th edge block corresponding to input switch $I_{3|V|+m}$ and output switch $O_{3|V|+m}$, $m \in [ \, | E(G) | \, ]$.
    
    The set of flows is denoted by $\mathcal{F}( G )$, and is composed of three types of flows: 
    
    \begin{itemize}
        \item For each vertex $v_k \in V(G)$, there is a set of \textit{vertex flows} (in blue) in the $k$'th vertex block corresponding to the translation of the cross gadget of size $3$ to that block.

        \item For each edge $e_m = \{ v_k, v_l \} \in E(G)$, there are two \textit{edge flows} in the $m$'th edge block (in black): one flow $( s^1_{3|V(G)|+m}, t^1_{3|V(G)|+m} )$ with demand $1$, and one flow $( s^2_{3|V(G)|+m}, t^2_{3|V(G)|+m} )$ with demand~$1$.
        \item For each edge $e_m = \{ v_k, v_l \} \in E(G)$, there is one \textit{incident flow} from the $k$'th vertex block to the $m$'th edge block, and one \textit{incident flow} from the $l$'th vertex block to the $m$'th edge block (in pink): one flow $( s^3_{3(k-1) + \text{rank}_{v_k}( v_l )}, t^3_{3 |V(G)|+m} )$ with demand $\nicefrac{1}{2}$ and one flow $( s^3_{3(l-1) + \text{rank}_{v_l}( v_k ) }, t^3_{3 |V(G)|+m} )$ with demand $\nicefrac{1}{2}$. 
    \end{itemize}  

    \begin{figure}
        \centering
        \captionsetup{labelfont=bf}
        \begin{subfigure}[t]{.49\textwidth}
            \centering
            \captionsetup{labelfont=bf}
            \includegraphics[width=0.6\textwidth]{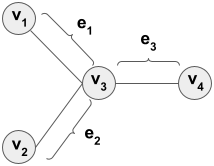} 
            \subcaption{Input graph $G$ to the $3$-edge \newline coloring problem.}
            \label{fig:limits_approximation_1}
        \end{subfigure}
        \begin{subfigure}[t]{.49\textwidth}
            \centering
            \captionsetup{labelfont=bf}
            \includegraphics[width=0.86\textwidth]{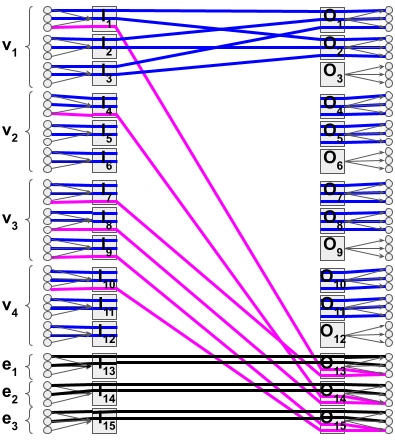} 
            \subcaption{Input flows $\mathcal{F}( G )$ to the minimum congestion routing problem.}
            \label{fig:limits_approximation_2}
        \end{subfigure}
        \caption{The reduction underlying the proof of Theorem~\ref{thm:limit_approximation} for a particular input graph. The vertex flows in the $2$'nd, $3$'rd, and $4$'th vertex blocks are not shown. The coloring of the different types of flows in $\mathcal{F}( G )$ is unrelated to the edge coloring in $G$.}
        \label{fig:limits_approximation}
    \end{figure}

    The objective of the construction is to establish the following correspondence between a routing of $\mathcal{F}(G)$ and an edge-coloring of $G$. Given a routing of $\mathcal{F}( G )$ with congestion $1$, the routing of the incident flows in $\mathcal{F}( G )$ yields a proper edge-coloring of $G$ using at most $3$ colors; conversely, given a proper edge-coloring of $G$ using at most $3$ colors, the coloring yields a routing of the incident flows in $\mathcal{F}(G)$ that can be extended to a routing of $\mathcal{F}(G)$ with congestion $1$. The next lemma paves the way for this correspondence by identifying the properties met by the routing of the incident flows in a routing of $\mathcal{F}(G)$ with congestion $1$.
        
    \begin{lemma}
        \label{lem:construction}
        For every routing of $\mathcal{F}(G)$ with congestion $1$, the routing of the incident flows satisfies the following properties:
        \begin{enumerate}[label=(\textbf{P\arabic*)}]
            \item For every edge $e_m \in E(G)$, the two incident flows entering the $m$'th edge block are assigned to the same middle switch.
            \item For every vertex $v_k \in V(G)$, each of the incident flows leaving the $k$'th vertex block is assigned to a different middle switch.
        \end{enumerate}
        If there is a routing of the incident flows in $\mathcal{F}(G)$ satisfying the previous properties, then there is a routing of $\mathcal{F}(G)$ with congestion $1$. 
    \end{lemma}
    
    \begin{proof}
        To show the first part, suppose that there is a routing of $\mathcal{F}( G )$ with congestion $1$. We prove that the first property holds. Since the routing assigns the two edge flows entering the $m$'th edge block to two different middle switches, we conclude that the two incident flows are assigned to the remaining middle switch. We now prove that the second property holds. From Lemma~\ref{lem:gadget}, we known that the routing does not assign a vertex flow to a different middle switch at each input switch in the $k$'th vertex block. Therefore, we conclude that the incident flows leaving the $k$'th vertex block are assigned to different middle switches.
        
        To show the second part, suppose that there is a routing of the incident flows in $\mathcal{F}( G )$ satisfying both properties. We prove that it can be extended to a routing of $\mathcal{F}( G )$ with congestion $1$. From the first property, the two edge flows in the $m$'th edge block can be routed with congestion $1$ if assigned to the two remaining middle switches. From the second property, the vertex flows leaving the $k$'th vertex block can be routed with congestion $1$ if assigned according to the elemental routing of the cross gadget (upon the appropriate numbering of the middle switches). The conclusion follows.        
    \end{proof}

    \textbf{Correctness of the reduction:} We show that there is a routing of $\mathcal{F}(G)$ with congestion $1$ if and only if there is a proper edge-coloring of $G$ using at most $3$ colors. 
    
    ($\Rightarrow$) Suppose that there is a routing of $\mathcal{F}(G)$ with congestion $1$. Consider the edge-coloring of $G$ using at most $3$ colors obtained from the routing where, for each edge $e_m \in E(G)$, edge $e_m$ is assigned color $c$ if the two incident flows entering the $m$'th edge block are assigned to middle switch $M_c$, for some $c \in \{ 1, 2, 3 \}$. From the first part of Lemma~\ref{lem:construction}, since the routing of the incident flows meets the first property, we deduce that the posited coloring is well-defined, and, since the routing of the incident flows meets the second property, we further deduce that it is proper, to conclude that there is a proper edge-coloring of $G$ using at most 3 colors.         
    
    ($\Leftarrow$) Suppose that there is a proper edge-coloring of $G$ using at most $3$ colors. Consider the routing of the incident flows in $\mathcal{F}( G )$ obtained from the coloring where, for each edge $e_m \in E(G)$, the two incident flows entering the $m$'th edge block are assigned to middle switch $M_c$ if $e_m$ is assigned color $c$, for some $c \in \{ 1, 2, 3 \}$. By construction, we deduce that the posited routing of the incident flows meets the first property, and, since the coloring is proper, we further deduce that it meets the second property, to conclude from the second part of Lemma~\ref{lem:construction} that there is a routing of $\mathcal{F}( G )$ with congestion $1$. \qedhere   
\end{proof}

%% file: 6-online_lower_bounds.tex
\quad We present lower bounds on the worst-case congestion and approximation of minimum congestion routings in Clos networks by online algorithms. An online algorithm is presented with a sequence of flows. A \textit{deterministic online algorithm} defines a routing for every sequence of flows that satisfies the following property: for all prefixes $P$ of a sequence $F$ of flows, the routing for $P$ when the algorithm is given $P$ equals the routing for $P$ when the algorithm is given $F$. A \textit{randomized online algorithm} defines a probability distribution over the set of all deterministic algorithms; a randomized algorithm with a single-point distribution reduces to a deterministic algorithm. 

In \S\ref{sec:deterministic}, we show that, for any deterministic online algorithm, there is a sequence of flows for which, while the congestion of a minimum congestion routing is $1$, the congestion of the routing returned by the algorithm is at least $2$. In \S\ref{sec:randomized}, we generalize the previous result from any deterministic online to any randomized online algorithm. These results implies that it is impossible for any online algorithm, randomized or deterministic, to approximate a minimum congestion routing by a factor less than $2$.

\subsection{Limits to Deterministic Online Algorithms}
\label{sec:deterministic}

\quad \enskip In this section, we assume that all flows have demand $1$, meaning that there is at most one flow per source and per destination. Under this assumption, the congestion of a minimum congestion routing is $1$ for every set of flows. In contrast, the theorem below shows no deterministic online algorithm can avoid returning a routing with congestion at least $2$ for some sequence of flows. 

\begin{theorem}
    \label{thm:online_lower_bound_congestion}
    Consider a Clos network $C_{N, R}$, for $N \ge 2$ and $R \ge 3$. For every deterministic online algorithm, there is a sequence of flows with demand $1$ for which the congestion of the routing returned by the algorithm is at least $2$.
\end{theorem}

\begin{proof}[\textbf{Proof of Theorem~\ref{thm:online_lower_bound_congestion}}]
    We design two sequences of flows in the Clos network $C_{N, 3}$ composed of flows with demand $1$ for which every deterministic online algorithm returns a routing with congestion at least $2$ for at least one of the sequences. As in the proof of Theorem~\ref{thm:limit_congestion}, if a sequence of flows satisfies the theorem requirements in the network $C_{N, 3}$, then the same sequence satisfies them in the network $C_{N, R}$.  

    \vspace{7pt}

    \textbf{Designing the sequences of flows.} The posited sequences are denoted by $X = (X_1, X_2)$ and $Y = (Y_1, Y_2)$, and are illustrated in Figure~\ref{fig:online_lower_bound_congestion} for $N = 4$ and $R = 3$. Each of the sequences consists of one subsequence followed by another, with the sequences agreeing on the prefix, $X_1 = Y_1$, and disagreeing on suffix, $X_2 \neq Y_2$. The arrival order among the flows within a subsequence is arbitrary, and each flow is identified by its input-output switch pair (with each flow incident to a common ToR switch incident to a different server of that switch).
    
    \begin{itemize}
        \item Subsequence $X_1$ consists of $\nicefrac{N}{2}$ flows ($I_1$, $O_1$) (in orange), and $\nicefrac{N}{2}$ flows ($I_2$, $O_2$) (in~green).
        \item Subsequence $X_2$ consists of $\nicefrac{N}{2}$ flows ($I_1$, $O_2$) (in blue).
        \item Subsequence $Y_2$ consists of $\nicefrac{N}{2}$ flows ($I_3$, $O_1$) (in pink) and $\nicefrac{N}{2}$ flows ($I_3$, $O_2$) (in black).
    \end{itemize}

    \begin{figure}[t]
        \centering
        \captionsetup{labelfont=bf}
        \begin{subfigure}{.49\textwidth}
            \centering
            \captionsetup{labelfont=bf}    
            \includegraphics[width=0.8\textwidth]{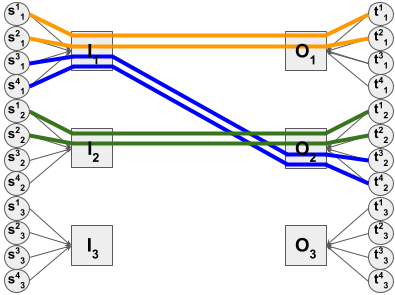} 
            \subcaption{Sequence $X$.}
            \label{fig:online_lower_bound_congestion_1}
        \end{subfigure}
        \begin{subfigure}{.49\textwidth}
            \centering
            \captionsetup{labelfont=bf}     
            \includegraphics[width=0.8\textwidth]{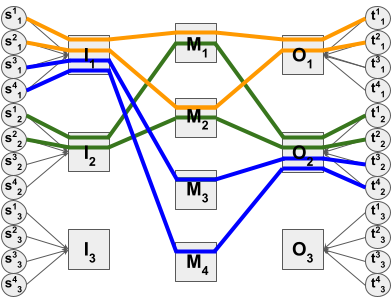} 
            \subcaption{Link-disjoint routing of $X$.}
            \label{fig:online_lower_bound_congestion_2}
        \end{subfigure}
    
        \vspace{5pt}
        
        \begin{subfigure}{.49\textwidth}
            \centering
            \captionsetup{labelfont=bf}       
            \includegraphics[width=0.8\textwidth]{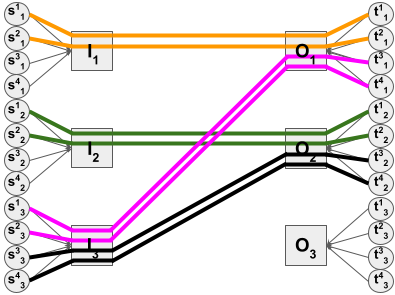} 
            \subcaption{Sequence $Y$.}
            \label{fig:online_lower_bound_congestion_3}
        \end{subfigure}
        \begin{subfigure}{.49\textwidth}
            \centering
            \captionsetup{labelfont=bf}       
            \includegraphics[width=0.8\textwidth]{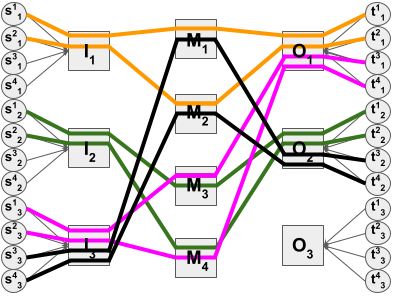} 
            \subcaption{Link-disjoint routing of $Y$.}
            \label{fig:online_lower_bound_congestion_4}
        \end{subfigure}

    \caption{The sequences of flows underlying Theorem~\ref{thm:online_lower_bound_congestion} for ${N = 4}$.}
    \label{fig:online_lower_bound_congestion}
    \end{figure}

    The key property of sequences $X$ and $Y$, detailed in the next lemma, is that, although $X_1 = Y_1$, for any link-disjoint routing of $X$ and any link-disjoint routing of $Y$, the routing of $X_1$ in the former differs from the routing of $X_1$ in the latter. Consequently, if the prefix of a given sequence is $X_1$ and the suffix is either $X_2$ or $Y_2$, then any deterministic online algorithm must choose between one of the previous routings of $X_1$ without knowing in advance whether $X_1$ will be followed by $X_2$ or by $Y_2$, such that the algorithm fails to returns a link-disjoint routing for at least one of $X$ and $Y$, as formalized in the final proof step.
    
    \begin{lemma}
        \label{lem:sequences} 
        Let $\mathcal{M}$ be the set of all middle switches in the Clos network, $\mathcal{M} \coloneq \{ M_i \, | \, i \in [ N ] \}$. For every link-disjoint routing of sequence $X$, it holds that:
        
        \begin{enumerate}[label=(\textbf{P\arabic*)}]
            \item Each of the $\nicefrac{N}{2}$ flows $(I_1, O_1)$ is assigned to a different middle switch from a set $\mathcal{N}$ of $\nicefrac{N}{2}$ middle switches; likewise, each of the $\nicefrac{N}{2}$ flows $(I_2, O_2)$ is assigned to a different middle switch from the same set $\mathcal{N}$. 
        \end{enumerate} 
        
        \noindent For every link-disjoint routing of sequence $Y$, it holds that:
        
        \begin{enumerate}[label=(\textbf{P\arabic*)}]
            \setcounter{enumi}{1}
            \item Each of the $\nicefrac{N}{2}$ flows $(I_1, O_1)$ is assigned to a different middle switch from a set $\mathcal{N}$ of $\nicefrac{N}{2}$ middle switches; oppositely, each of the $\nicefrac{N}{2}$ flows $(I_2, O_2)$ is assigned to a different middle switch from the set $\mathcal{M} \setminus \mathcal{N}$.
        \end{enumerate}        
        
        \noindent Therefore, no routing of $X_1$ can satisfy both properties. 
    \end{lemma}

    \begin{proof}
        To show the first statement, consider an arbitrary link-disjoint routing of $X$. Since~the $N$ flows leaving $I_1$ are spread over all $N$ middle switches, a middle switch is assigned a flow $(I_1, O_1)$ precisely when it is not assigned a flow $(I_1, O_2)$; similarly, since the $N$ flows entering $O_2$ are spread over all $N$ middle switches, a middle switch is assigned a flow $(I_2, O_2)$ precisely when it is not assigned a flow $(I_1, O_2)$. Therefore, each of the $\nicefrac{N}{2}$ middle switches to which a flow $(I_1, O_1)$ is assigned is also assigned a flow $(I_2, O_2)$, thus proving the first statement. 
    
        To show the second statement, consider now an arbitrary link-disjoint routing of $Y$. Following the same argument as above, a middle switch is assigned a flow $(I_1, O_1)$ precisely when it is not assigned a flow $(I_3, O_1)$, and a middle switch is assigned a flow $(I_3, O_2)$ precisely when it is not assigned a flow $(I_3, O_1)$; consequently, a middle switch is assigned a flow $(I_1, O_1)$ precisely when it is assigned a flow $(I_3, O_2)$. Since a middle switch is assigned a flow $(I_3, O_2)$ precisely when when it is not assigned a flow $(I_2, O_2)$, each of the $\nicefrac{N}{2}$ middle switches to which a flow $(I_1, O_1)$ is not assigned assigned a flow $(I_2, O_2)$, thus proving the second statement. 
    \end{proof}

    \textbf{Calculating the congestion of a deterministic algorithm.} Consider an arbitrary deterministic online algorithm. We show that the algorithm fails to return a link-disjoint routing for at least one of $X$ and $Y$, to conclude that it returns a routing with congestion at least $2$ for at least one of them. Let $r$ be the routing returned by the algorithm for prefix $X_1$. We distinguish two cases, depending on whether $r$ satisfies property P1 in Lemma~\ref{lem:sequences}.
    
    \begin{itemize}
        \item[$\#$] \textit{Case 1:} Suppose that $r$ satisfies property P1, in which case $r$ does not satisfy property P2. Then, we deduce that the algorithm returns a routing for $Y$ that does not property P2, which means that the routing for $Y$ is not link-disjoint. 
        \item[$\#$] \textit{Case 2:} Suppose that $r$ does not satisfy property P1. Then, we deduce that the algorithm returns a routing for $X$ that does not satisfy property P1, which means that the routing for $X$ is not link-disjoint. \qedhere
    \end{itemize}
\end{proof}

\subsection{Limits to Randomized Online Algorithms}
\label{sec:randomized}

\quad \enskip The previous theorem does not preclude the possibility that a randomized online algorithm can avoid with non-negligible probability returning a routing with congestion greater than $2$ for every sequence of flows. The next theorem refutes this possibility.  

\begin{theorem}
    \label{thm:online_lower_bound_congestion_random}
    Consider a Clos network $C_{N, R}$, for $N \ge 2$ and $R \ge 3$. For every randomized online algorithm, there is a sequence of flows with demand $1$ for which the expected congestion of the routing returned by the algorithm is at least $2 - \nicefrac{1}{2^S}$, where $S = \lfloor \nicefrac{R}{3} \rfloor$.
\end{theorem}

\begin{proof}[\textbf{Proof of Theorem~\ref{thm:online_lower_bound_congestion_random}}]
    The proof of the theorem makes use of Yao's Minimax Principle, which is recalled below in the context of the minimum congestion routing problem in Clos networks. Let $\mathscr{S}$ be the set of all sequences of flows, and $\mathscr{A}$ be the set of all deterministic algorithms. Given a probability distribution $p$ over $\mathscr{S}$ and a probability distribution $q$ over $\mathscr{A}$, denote by $\mathcal{S}_p$ the random sequence drawn from $p$, and by $A_q$ the randomized algorithm drawn from $q$. We write $E[ A_q( \mathcal{S} ) ]$ for the expected congestion of $A_q$ for a sequence $\mathcal{S}$, and $E[ A( \mathcal{S}_p ) ]$ for the expected congestion of a deterministic algorithm $A$ for $\mathcal{S}_p $.
    
    \begin{lemma}[Yao's Minimax Principle~\cite{Yao_1977, Motwani_1995}]
        For every randomized algorithm $A_q$ and every random sequence $\mathcal{S}_p$, there is a sequence $\mathcal{S}$ of flows for which the expected congestion of $A_q$ for $\mathcal{S}$ is at least the expected congestion of the optimal deterministic algorithm for $\mathcal{S}_p$:
        \begin{align*}
            \max_{ \mathcal{S} \in \mathscr{S} } E[ A_q( \mathcal{S} ) ] \, \ge \, \min_{A \in \mathcal{A}} E[ A( \mathcal{S}_p  ) ].
        \end{align*}
    \end{lemma}

    Therefore, we design a random sequence based on the sequences $X$ and $Y$ introduced earlier for which the expected congestion of an optimal deterministic algorithm for that random sequence is at least $2 {-} \nicefrac{1}{2^S}$. 

    \vspace{7pt}

    \textbf{Designing the random sequence of flows.} The posited random sequence is denoted by $\mathcal{S}_p$ and is drawn from the uniform distribution $p$ over the following set of $2^S$ sequences. The set of $2^S$ sequences is denoted by $\{ \, \mathcal{S}_i \, | \, i \in [ 2^S ] \, \}$, with $\mathcal{S}_i \coloneq ( \, \mathcal{S}^j_i \, | \, j \in [ S ] \, )$. Consider that the ToR switches are divided into $S$ blocks, with the $j$'th block designating input switches $I_{3(j-1)+1}$ through $I_{3(j-1)+3}$ and output switches $O_{3(j-1)+1}$ through $O_{3(j-1)+3}$, $j \in [S]$. Denote by $X_j$ and $Y_j$, respectively, the translation of $X$ and $Y$ to the $j$'th block. Each of the $2^S$ sequences $S_i$ is composed of $S$ subsequences, with each of the subsequences $S^j_i$ equal to $X_j$ or $Y_j$. The arrival order among the subsequences within a sequence is arbitrary. Let $< b_j( i ) \, | \, j \in [ S ] >$ be the binary representation of an integer $i \in [0, 2^S {-} 1]$.
    
    \begin{itemize}
        \item If $b_j(i-1) = 0$, then $\mathcal{S}^j_i = X_j$; otherwise, $\mathcal{S}^j_i = Y_j$, for all $i \in [2^S]$ and $j \in [S]$.
    \end{itemize}

    The key property of the set $\{ \, \mathcal{S}_i \, | \, i \in [ 2^S ] \, \}$, detailed in the next lemma, is that no online deterministic algorithm can return a link-disjoint routing for more than one sequence in the set. Consequently, every deterministic online algorithm returns a routing with congestion at least $2$ for at least $2^S{-}1$ of the $2^S$ sequences, as formalized in the final proof step.
    
    \begin{lemma}
        \label{lemma:supersequences}   
        Every deterministic online algorithm returns a link-disjoint routing for at most one sequence in the set $\{ \mathcal{S}_i \, | \, i \in [ 2^S ] \}$.
    \end{lemma}

    \begin{proof}
        We show that for every pair of sequences in the set $\{ \mathcal{S}_i \, | \, i \in [ 2^S ] \}$ no deterministic online algorithm can return a link-disjoint routing for both sequences. Consider an arbitrary deterministic online algorithm and two sequences $\mathcal{S}_{i_1}, \mathcal{S}_{i_2} \in \{ \mathcal{S}_i \, | \, i \in [ 2^S ] \}$ such that $\mathcal{S}_{i_1} \neq \mathcal{S}_{i_2}$, meaning that $\mathcal{S}^j_{i_1} \neq \mathcal{S}^j_{i_2}$ for some $j \in [S]$. From Lemma~\ref{lem:sequences}, the algorithm cannot return a link-disjoint routing for both $\mathcal{S}^j_{i_1}$ and $\mathcal{S}^j_{i_2}$, implying that it cannot return a link-disjoint routing for both $\mathcal{S}_{i_1}$ and $\mathcal{S}_{i_2}$. 
    \end{proof}

    \pagebreak

    \textbf{Calculating the expected congestion of a deterministic algorithm.} Consider an arbitrary deterministic online algorithm $A$. From Lemma~\ref{lemma:supersequences}, we write 
    \begin{align*}
        E[ A( \mathcal{S}_p  ) ] & \ge \frac{1}{2^S} \times 1 + (1 - \frac{1}{2^S}) \times 2 \\
                                 &   = 2 - \frac{1}{2^S}, 
    \end{align*}    
    to deduce that the expected congestion of an optimal deterministic algorithm for $\mathcal{S}_p$ is at least $2 {-} \nicefrac{1}{2^S}$. Consequently, from Yao's Minimax Principle, we conclude that for every randomized algorithm there is a sequence of flows for which the expected congestion is at least $2 {-} \nicefrac{1}{2^S}$.
\end{proof}

\begin{corollary}
    \label{cor:online_lower_bound_congestion}
    Consider a Clos network $C_{N, R}$, for $N \ge 2$ and $R \ge 3$. For every $c < 2 - \nicefrac{1}{2^S}$, where $S = \lceil \nicefrac{R}{3} \rceil$, there is no $c$-approximation randomized online algorithm for the minimum congestion routing problem in Clos networks.
\end{corollary}

It not hard to show that, if every flow has demand $1$, then the Unsorted Greedy algorithm returns a routing with congestion at most $2$, and thus approximates a minimum congestion routing by a factor of $2$. Hence, in this special case, the lower bound of Corollary~\ref{cor:online_lower_bound_congestion} is tight.

%% file: 7-related_work.tex
\quad \textbf{Formal analysis of minimum congestion routing.} Chiesa, Kindler, and Shapira~\cite{Chiesa_2017} present a formal analysis of minimum congestion routings of unsplittable flows in Clos networks. However, their setting differs from ours in that they assume that the topology modeling a data-center is a \textit{folded} Clos network rather than an \textit{unfolded} Clos network, which is the more accurate representation considered in our work.
\unskip\footnote{An unfolded Clos network is obtained from a folded Clos network by mirroring the latter around the middle switches, and establishing a direction for traffic forwarding, with the capacity of each link in the unfolded network half the capacity of the corresponding link in the folded network. Since in almost every wired network deployed data transmission is independent between the two directions of a link~\cite{Pfister_2001, Desanti_2006, Birrittella_2015, ethernet}, the unfolded Clos network is the more accurate representation of a data-center.}
Consequently, their results are more pessimistic than ours. First, they show that for a set of flows with unit demands deciding whether there is a routing with congestion equal to $1$ is NP-complete, and thus conclude that it is impossible to approximate a minimum congestion routing by a factor less than $2$; with unfolded Clos networks, it is known that this problem is in $P$~\cite{Hwang_1983}, and it can be shown that both the Sorted Greedy and the Melen-Turner algorithm approximate a minimum congestion routing by a tight factor of $2$. Second, they give a local-search algorithm that, while there is a flow and a path such that re-assigning that flow to that path decreases the congestion of the routing, re-assigns the flow; with unfolded Clos networks, it can be shown that their algorithm approximates a minimum congestion routing by a tight factor of $3$. \newline

\textbf{Multirate rearrangeability problem.} The multirate rearrangeability problem in Clos networks is the closest problem to the minimum congestion routing problem. The setting consists of a Clos network whose number of servers per ToR switch and of ToR switches is fixed but whose number of middle switches is variable, and a set of flows with demands limited by the unit capacity of links between servers and ToR switches. The goal is to find a routing with congestion at most $1$ while minimizing the number of middle switches used.

The best known algorithm is given by Khan and Singh~\cite{Khan_2015}, and shows that it is possible to route every set of flows with congestion at most $1$ using at most $\lceil \nicefrac{20}{9} \, N \rceil$ middle switches, where $N$ is the number of servers per ToR switch. The algorithm fixes $\lceil \nicefrac{20}{9} \, N \rceil$ middle switches, and routes a set of flows in two phases. The first phase considers a subset of the flows with demand at least $\nicefrac{1}{10}$, and finds a link-disjoint routing for these flows. The second phase sorts the flows in decreasing order of demands, and assigns each flow to an arbitrary middle switch for which the congestion does not exceed $1$. Our algorithm draws inspiration from the two-phase approach of the Khan-Singh algorithm. The best known lower bound is given by Ngo and Vu~\cite{Ngo_2003}, and shows that there are sets of flows for which every routing with congestion $1$ uses at least $\lceil \nicefrac{5}{4} \, N \rceil$ middle switches. While the underlying construction yields a lower bound of $\nicefrac{6}{5}$ on worst-case congestion, we improve this bound to $\nicefrac{3}{2}$.

%% file: 8-discussion.tex
\quad Today, the commodification of network capacity has made congestion a key consideration~in the design of data-center topologies and routing algorithms. However, under the realistic constraint of flow unsplittability, existing algorithms provide only weak worst-case bounds on congestion. Our work presents the first characterization of upper and lower bounds~on~the minimum congestion achievable in Clos networks with unsplittable flows. In particular,it provides upper bounds of $\nicefrac{9}{5}$ and lower bounds of $\nicefrac{3}{2}$ on the worst-case congestion~and approximation by polynomial-time algorithms, thus breaking through the barrier of~$2$ established by existing heuristics, and raising the bar significantly away from $1$. It also provides lower bounds of $2$ on the worst-case congestion and approximation by online algorithms, thus separating the minimum congestion achievable in the offline and online settings.

Our work leaves several open questions. In the offline setting, the main ask is to eliminate the discrepancy with respect to congestion and approximation between the $\nicefrac{9}{5}$ factor yielded by the new routing algorithm and the $\nicefrac{3}{2}$ factor yielded by the new lower bounds, the latter of which we believe is tight. While the analysis of the new algorithm can be improved, arriving at the $\nicefrac{3}{2}$ factor should require new techniques. Similarly to the offline setting, in the online setting, the key question is to narrow the gap with respect to congestion and approximation between the $2$ factor yielded by the new lower bounds and the $3$ factor yielded by the Unsorted Greedy algorithm. Moving beyond unsplittable flows, it is plausive that future data-centers deploy transport protocols that implement some degree of flow splittability. In this context, it is important to understand the trade-off between the number of paths over which each flow can be split and the minimum congestion achievable.